%% file: ms.tex
\def\tpdp{0}

\documentclass[a4paper]{article}

\usepackage[english]{babel}
\usepackage[utf8x]{inputenc}
\usepackage[T1]{fontenc}

\usepackage{preamble}

\title{Differentially Private Histograms \\ in the Shuffle Model from Fake Users}
\author{Albert Cheu\thanks{This work was done while the author was a PhD. student at Northeastern University. Email ac2305@georgetown.edu or cheu.a@northeastern.edu.}~ and Maxim Zhilyaev\thanks{Mindstrong Inc. Email maxim.zhilyaev@gmail.com}}

\begin{document}
\maketitle

\ifnum\tpdp=0
    \begin{abstract}
        There has been much recent work in the shuffle model of differential privacy, particularly for approximate $d$-bin histograms. While these protocols achieve low error, the number of messages sent by each user---the message complexity---has so far scaled with $d$ or the privacy parameters. The message complexity is an informative predictor of a shuffle protocol's resource consumption. We present a protocol whose message complexity is two when there are sufficiently many users. The protocol essentially pairs each row in the dataset with a fake row and performs a simple randomization on all rows. We show that the error introduced by the protocol is small, using rigorous analysis as well as experiments on real-world data. We also prove that corrupt users have a relatively low impact on our protocol's estimates.
    \end{abstract}
\fi

\input{intro}

\ifnum\tpdp=0
    
    \input{prelims}
    
    \input{protocol}

    \input{reduce_comm}
    
    \input{experiments}
\else
    \input{tpdp_techniques}
\fi

\section*{Acknowledgements}
We would like to thank Kobbi Nissim, Rasmus Pagh, and Jonathan Ullman for discussion and insight for the count-min analysis.

\bibliographystyle{plain}
\bibliography{refs}

\ifnum\tpdp=0
    \appendix
    \input{appendix/technical}
    \input{appendix/attack}
    \input{appendix/amplification}
\fi

\end{document}

%% file: intro.tex
\section{Introduction}

\ifnum\tpdp=1
    Because devices are resource constrained, a keyboard app offers word corrections from a smaller pool than the entire vocabulary. To obtain a list of the most common words, user devices could participate in a differentially private computation that estimates word \emph{histograms}. The app developer can then approximate the top-$t$ words for some relatively small $t$. There exists an $(\eps, \delta)$-\emph{centrally private} algorithm that computes $d$-bin histograms from $n$ users up to maximum ($\ell_\infty$) error $O( \tfrac{1}{\eps n} \log \tfrac{1}{\delta} )$ \cite{BNS16}. Conversely, Bassily \& Smith show $(\eps,o(1/n))$-\emph{local privacy} incurs a maximum error of $\Omega(\tfrac{1}{\eps} \sqrt{ \tfrac{\log d}{n} } )$ \cite{BS15}.
\else
    Given that statistical computations often involve data sourced from human users, an analyst could execute differentially private algorithms in the central model (also called \emph{centrally private algorithms}). Originally defined by Dwork, McSherry, Nissim, and Smith \cite{DMNS06}, these algorithms provide quantifiable protection to data contributors at a small price in terms of accuracy. As an example, there exists an $(\eps,\delta)$-centrally private algorithm that computes $d$-bin histograms from $n$ users up to maximum ($\ell_\infty$) error $O( \tfrac{1}{\eps n} \log \tfrac{1}{\delta} )$ \cite{BNS16}.
    
    We focus on computing accurate histograms since they allow approximate top-$t$ selection, the set of $t$ data values that occur most frequently in a population. One application is smart-phone autocomplete. Because devices are resource constrained, a keyboard offers word corrections from a smaller pool than the entire vocabulary. To obtain a list of the most common words, user devices could participate in a differentially private computation that estimates word frequencies.
    
    Users contributing to a centrally private algorithm need to trust that the analyst correctly executes the algorithm and does not leak their data. To collect data from less trusting users, analysts can instead implement \emph{locally private protocols}: each user applies a differentially private algorithm on their data and sends a message containing the algorithm's output to the analyst. This weaker trust assumption comes at a price: there are lower bounds that show locally private protocols have significantly more error than the centrally private counterparts. Returning to the histogram example, Bassily \& Smith show $(\eps,o(1/n))$-local privacy incurs a maximum error of $\Omega(\tfrac{1}{\eps} \sqrt{ \tfrac{\log d}{n} } )$ \cite{BS15}.
\fi

Originating with work by Bittau et al. and Cheu et al. \cite{BEMMR+17,CSUZZ19}, \emph{shuffle privacy} has emerged as an appealing middle-ground. Here, we assume that there is a service called the shuffler that uniformly permutes user messages. The output of the shuffler must satisfy $(\eps,\delta)$-differential privacy. Intuitively, if each user generates a locally private message, then the anonymity provided by the shuffler ``amplifies'' the privacy guarantees.

But a user can send \emph{multiple} messages to the shuffler. And a users does not need to produce these messages in a differentially private manner, since we only require that the output of the shuffler is differentially private. This flexibility is leveraged by the histogram protocol of Balcer \& Cheu \cite{BC20}, where each user sends $d+1$ messages and the maximum error is $O(\tfrac{1}{\eps^2} \log \tfrac{1}{\delta})$ for $\delta = O(1/n)$. Alternative histogram protocols in the shuffle model have been introduced by Ghazi, Golowich, Kumar, Pagh, and Velingker \cite{GGKPV19} and by Ghazi, Kumar, Manurangsi, and Pagh \cite{GKMP20}. As shown in Table \ref{tab:comparison}, these protocols demand much fewer messages from each user than the protocol from \cite{BC20}.

We will use \emph{message complexity} to refer to the number of messages sent by each user and \emph{communication complexity} to refer to total number of bits consumed by those messages. The message complexity is necessary to have a complete picture of a protocol's resource consumption. For starters, the amount of randomness needed to perform the shuffle is a function of the message complexity but not the length of each message. Furthermore, two protocols $\cP,\cP'$ with the same communication complexity can incur different costs, since the physical delivery of a message over a network in a secure fashion requires overhead. If $\cP$ sends more messages than $\cP'$, the computing cost of transmitting messages is larger for $\cP$, since it needs to perform cryptographic operations on each message. Also, the bandwidth overhead is larger for $\cP$, due to both encryption and physical network protocols such as TCP/IP.

%[old paragraph] Furthermore, two protocols can have the same communication complexity on paper but consume different amounts of bandwidth in practice. This phenomenon occurs when the shuffler is designed to permute objects of fixed size, like an off-the-shelf onion router that operates on network packets. When $d$ bits can fit in one such packet, sending one $d$-bit message is preferable to $d$ one-bit messages.

In light of the above, one can ask the following question:
\begin{quote}
\textit{Are there shuffle private protocols for histograms that have low message complexity but still provide estimates that are competitive with prior work?}
\end{quote}

Given the distributed nature of local and shuffle protocols, they are impacted by users who deviate from the intended behavior. In the local privacy literature, there is research on \emph{manipulation attacks} where corrupted users aim to skew estimates and tests by sending carefully crafted messages. One baseline attack is to simply feed wrong inputs into the protocol, but the prior work has shown that there are attacks against locally private protocols that introduce significantly worse error (see e.g. Cao, Jia, and Gong \cite{CJG19} and Cheu, Smith, and Ullman \cite{CSU19} and citations within). Here, we investigate manipulation against shuffle private protocols. Specifically,
\begin{quote}
\textit{Are there shuffle private protocols for histograms that are robust to manipulation?}
\end{quote}

%The result implies that no adversary of the analyzer (e.g. rival app developer) can hope to introduce too much more error to our protocol than simply giving wrong inputs. Specifically,

\subsection{Our Contributions}
Our primary contribution is a shuffle private protocol for  histograms that answers both questions in the affirmative. For a large range of $n$, the communication complexity is the same as \cite{BC20} up to a logarithmic factor but the message complexity can be as small as two. For a natural use case and set of parameters, experiments also show that the new protocol is more accurate than \cite{BC20}. Finally, we show that one consequence of the low message complexity is robustness to manipulation by corrupt users.

\ifnum\tpdp=0 Section \ref{sec:protocol} \else The full version of this work \fi contains the full specification and analysis, but we give an overview of the main features in the theorem below.

%old version
\begin{comment}
    \begin{thm}[Informal]
    \label{thm:informal}
    For any number of users $n$ and privacy parameters $\eps = O(1)$, $\delta < 1/100$, there is an $(\eps,\delta)$-differentially private shuffle protocol that approximates $d$-bin histograms with the following properties
    \begin{enumerate}[i.]
        \item Each user sends $k+1$ messages, where $k$ is any integer $\geq \lceil \frac{1}{n}\cdot O\left( \frac{1}{\eps^2} \log \frac{1}{\delta} +  \log d\right) \rceil$.
        \item The maxmimum error of any bin estimate is $f(k) \cdot O(\tfrac{\log d}{n} + \tfrac{1}{\eps n} \sqrt{\log d \log \tfrac{1}{\delta}  })$ with probability $9/10$,\footnote{We will consistently use $9/10$ success probability for simplicity, but it can be changed to an arbitrary constant.} where $f(k)$ monotonically decreases to 1.
        \item $m$ corrupted users can skew an estimate by at most $\tfrac{m}{n} \cdot (k+1)\cdot f(k)$
    \end{enumerate}
    \end{thm}
\end{comment}
\begin{thm}[Informal]
\label{thm:informal}
For any privacy parameters $\eps = O(1)$, $\delta < 1/100$, and number of users $n = \Omega(\log d + \tfrac{1}{\eps^2} \log \tfrac{1}{\delta})$, there is an $(\eps,\delta)$-differentially private shuffle protocol that approximates $d$-bin histograms with the following properties
\begin{enumerate}[i.]
    \item The message complexity is $k+1$, where $k$ can be set to any positive integer. Each message is $d$ bits.
    \item The maxmimum error of any bin estimate is $f(k) \cdot O(\tfrac{\log d}{n} + \tfrac{1}{\eps n} \sqrt{\log d \log \tfrac{1}{\delta}  })$ with probability $9/10$,\footnote{We use $9/10$ as a target success probability throughout this work, but it can be changed to any other constant without affecting the asymptotic analysis.} where $f(k)$ monotonically approaches 1 from above.
    \item $m$ corrupted users can skew an estimate by at most $\tfrac{m}{n} \cdot (k+1)\cdot f(k)$.
\end{enumerate}
\end{thm}

We unpack this theorem. Parts \textit{i} and \textit{ii} show that the protocol allows for a tradeoff between message complexity and the measurement accuracy, since increasing $k$ reduces the scaling factor $f(k)$. This may not be significant for large $n$, but it could be useful for smaller $n$ (e.g. the target population of a health survey can consist of much fewer subjects than the dictionary-building example). Re-scaling $k$ by a factor of $c$ will naturally increase the transmission cost by $c$ but the traffic remains feasible since $n$ is small. Thus, we can improve accuracy without altering the privacy guarantee.

Meanwhile, Part \textit{iii} bounds the impact of any manipulation attack. Each corrupt user in our protocol can introduce bias $O(\tfrac{1}{n})$ whenever $k\cdot f(k)=O(1)$. For comparison, we also prove that a protocol by Ghazi et al. \cite{GGKPV19} suffers bias $\Omega(\tfrac{1}{n} \cdot \tfrac{1}{\eps^2} \log\tfrac{1}{\delta})$ per corrupt user.

\medskip

Our other results build upon this protocol. \ifnum\tpdp=0 In Section \ref{sec:reduce-comm}, \else First, \fi we describe how to exponentially reduce the protocol's communication complexity. The price is an increased message complexity and a mildly increased error. \ifnum\tpdp=0 In Section \ref{sec:experiments}, \else Then, \fi we simulate our protocol on text sampled from Twitter. The error introduced by our protocol to the histogram is consistent with our theoretical bounds. We also show that the top-$t$ items in the output of the protocol are consistent with those in the raw dataset, for several choices of $t$. The experimental results of our protocol compare favorably to that of \cite{BC20}.

\ifnum\tpdp=0 Appendix \ref{apdx:amplification} presents \else Finally, the appendix of our full paper presents \fi an analysis of our main protocol in the special case where $k=0$. This is done by enhancing work by Ghazi et al. \cite{GGKPV19} with the state-of-the art amplification lemma by Feldman, McMillan, and Talwar \cite{FMT20}. The protocol's maximum error is now proportional to $1/n^{3/4}$ instead of $1/n$.

\ifnum\tpdp=0
    \paragraph{Techniques}
    Each user in our main protocol first encodes their data as a binary string with a single 1 bit. They then flip each bit independently with some fixed probability $q$. Next, they create $k$ other zero vectors and repeat this bit flipping, which corresponds to introducing $k$ fake users with null data. We show how to choose $q$ so that the $nk$ messages from these fake users provide differential privacy for the actual users. The privacy amplification lemma from \cite{FMT20} lets us analyze the case where $k=0$. The analyzer simply de-biases and adjusts the scale of the sums over messages.
    
    Our technique to reduce communication complexity proceeds in two stages. We first make the simple observation that a binary string with known length is equivalent to a list of the indices where the string has value 1. By construction, a message generated by our local randomizer is a binary string where the number of such indices has expectation $O(d q)$. Our choice of $q$ is proportional to $1/n$, so this alternative representation is very effective when $n$ approaches or exceeds $d$.
    
    The small $n$ regime motivates a second round of compression. We describe an adaptation of the count-min sketching technique. Given a uniformly random hash function, we can reduce the size of the domain $d$ to some $\hat{d}$ at the cost of some collisions. We repeatedly hash in order to reduce the likelihood of error due to collisions and run our histogram protocol on the hashed data. We remark that Ghazi et al. \cite{GGKPV19} build a specific histogram protocol out of count-min, while we use it as a tool that can improve the communication complexity of \emph{arbitrary} histogram protocols.
\fi

\subsection{Related Work}
Cheu, Smith, Ullman, Zeber, and Zhilyaev \cite{CSUZZ19} rigorously define the shuffle model and give a histogram protocol that requires $d$ messages per user. Balcer \& Cheu \cite{BC20} give a different protocol with the same message complexity (up to constants) but with maximum error independent of $d$. Because the tradeoff between error and message complexity in \cite{BC20} dominates that of \cite{CSUZZ19}, we omit the latter from Table \ref{tab:comparison}.

Ghazi et al. \cite{GGKPV19} propose multi-message shuffle protocols for histograms. These adapt the Hadamard response and Count-Min techniques from the local privacy and sketching literature. \cite{GGKPV19} also presents a single-message shuffle protocol, using the amplification lemma from Balle, Bell, Gasc{\'{o}}n, and Nissim \cite{BBGN19}. Unlike \ifnum\tpdp=0 Theorem \ref{thm:amplification-variant}\else our work\fi, their result does not give explicit constants and holds for a narrower range of $\eps,\delta$.

In follow-up work Ghazi et al. \cite{GKMP20} give a protocol where the message complexity shrinks as $n$ increases. Our protocol has the same property but at a faster rate. Specifically, our message complexity is two when $n$ is logarithmic in $d$ while the prior work requires $n$ to be linear in $d$.

\begin{table}
    \centering
    \begin{tabular}{cccc}
         Source & Bits per message & Messages per user & Max Error \\
         & & & {\footnotesize (90\% Confidence)} \\ \hline
         \rule{0ex}{3.5ex}\vspace{1ex} \cite{BC20} &  $O(\log d)$ & $d+1$ & $O(\tfrac{1}{\eps^2 n} \log \tfrac{1}{\delta}  )$ \\ \hline
         %%%%
         \rule{0ex}{3.5ex}\vspace{1ex}  & $O(\log d)$ & $O(d^{1/100})$ & $O\left(\frac{1}{\eps n}\sqrt{\log d \log \frac{1}{\delta}} \right)$ \\
         \rule{0ex}{3.5ex}\vspace{1ex} \cite{GGKPV19} & $O(\log n + \log \log d)$ &$O\left(\frac{\log^3 d}{\eps^2} \log \frac{\log d}{\delta} \right)$ & $O\left( \frac{\log^{3/2} d}{\eps n} \sqrt{ \log \frac{\log d}{\delta} } \right)$\\
         \rule{0ex}{3.5ex}\vspace{1ex} & $O(\log n \log d)$ & $O\left( \frac{1}{\eps^2} \log \frac{1}{\eps \delta}\right) $ & $O\left( \frac{\log d}{n} + \frac{1}{\eps n}\sqrt{\log d \log \frac{1}{\eps \delta}} \right)$ \\ \hline
         %%%%
         \rule{0ex}{3.5ex}\vspace{1ex} \cite{GKMP20} & $O(\log d)$ & $1+O(\tfrac{d}{n} \cdot u(\eps,\delta))$~$*$ & $O(\tfrac{1}{\eps n} \log d)$ \\ \hline
         %%%%
         \rule{0ex}{3.5ex}\vspace{1ex} {\footnotesize \ifnum\tpdp=0 Thm \ref{thm:maximum}\else This work\fi}  & $d$ & \multirow{2}{*}{2} & \multirow{2}{*}{$O\left( \frac{\log d}{n} + \frac{1}{\eps n}\sqrt{\log d \log \frac{1}{\delta}} \right)$} \\ \cline{1-2}
         \rule{0ex}{3.1ex}\vspace{1ex} {\footnotesize \ifnum\tpdp=0 Thm \ref{thm:replacement} \else This work\fi} & $O(\log d (1+ \tfrac{d}{n} \cdot v(\eps,\delta,d) ) )$ $*$ & & \\ \hline
         \rule{0ex}{3.5ex}\vspace{1ex} {\footnotesize \ifnum\tpdp=0 Thm \ref{thm:count-min} \else This work\fi}  & $O\left( \tfrac{1}{\eps^2 }  \log d \log^3\tfrac{\log d}{\delta} \right)$ $*$ & $O(\log d )$ & $O\left( \frac{1}{\eps n} \sqrt{\log d } \log^{3/2} \left( \frac{\log d}{\delta} \right) \right)$ \\ \hline
         %%%%
         \rule{0ex}{3.5ex}\vspace{1ex} {\footnotesize \ifnum\tpdp=0 Thm \ref{thm:amplification-variant}\else This work\fi, }  & \multirow{2}{*}{$d$} & \multirow{2}{*}{1} & \multirow{2}{*}{$O\left( \frac{\log d}{n} + \frac{\sqrt{\log d }}{\sqrt{\eps} n^{\tfrac{3}{4}}} \left(\log \frac{1}{\delta} \right)^{\tfrac{1}{4}} \right)$} \\
         {\footnotesize via \cite{GGKPV19, FMT20}} & & &\\ %\hline
    \end{tabular}
    \caption{Summary of shuffle protocols for histograms. To simplify presentation, we assume $\eps = O(1)$, $\delta= O(1/n)$, and $n = \Omega (\tfrac{\log d}{\eps^2} \log \tfrac{1}{\delta} \log \tfrac{\log d}{\delta} )$. We also use   $u(\eps,\delta)$ as shorthand for $\tfrac{\log^2(1/\delta)}{\eps^2}$ and $v(\eps,\delta,d)$ for $\log d + \tfrac{\log(1/\delta)}{\eps^2}$. $*$ indicates bounds on expected values.}
    \label{tab:comparison}
\end{table}

Manipulation attacks have previously been studied in the context of local privacy. Ambainis, Jakobsson, and Lipmaa \cite{AJL04} as well as Moran and Naor \cite{MN06} study the vulnerability of randomized response to these attacks. Work by Cao, Jia, and Gong \cite{CJG19} also consider attacks against histogram and heavy hitter protocols. Cheu, Smith, and Ullman \cite{CSU19} show that powerful attacks are inevitable for any locally private protocol. In particular, these attacks are stronger when the privacy guarantee is stronger or the data dimension is larger.

%% file: prelims.tex
\section{Preliminaries}
\subsection{Differential Privacy}
We define a dataset $\vec{x} \in \cX^n$ to be an ordered tuple of $n$ rows where each row is drawn from a data universe $\cX$ and corresponds to the data of one user. Two datasets $\vec{x},\vec{x}\,' \in \cX^n$ are considered \emph{neighbors} (denoted as $\vec{x} \sim \vec{x}\,'$) if they differ in at most one row.

\begin{definition}[Differential Privacy \cite{DMNS06}]
An algorithm $\cM: \cX^n \rightarrow \cZ$ satisfies \emph{$(\eps, \delta)$-differential privacy} if, for every pair of neighboring datasets $\vec{x}$ and $\vec{x}\,'$ and every subset $Z \subset \cZ$,
\begin{equation}
\label{eq:dp}
	\pr{}{\cM(\vec{x}\vphantom{'}) \in Z} \le e^\eps \cdot \pr{}{\cM(\vec{x}\,') \in Z} + \delta.    
\end{equation}
\end{definition}

We remark that an algorithm can be well-defined for a superset of the intended data universe $\overline{\cX} \supset \cX$ but \eqref{eq:dp} may not hold for every $\vec{x}\sim \vec{x}\,' \in \overline{\cX}$; in these cases, we will disambiguate by saying it satisfies \emph{differential privacy for inputs from $\cX$}.

Because this definition assumes that the algorithm $\cM$ has ``central'' access to compute on the entire raw dataset, we sometimes call this \emph{central} differential privacy. Two properties about differentially private algorithms will be useful. First, privacy is preserved under post-processing. 

\begin{fact}
\label{fact:post}
	For $(\eps,\delta)$-differentially private algorithm $\cM : \cX^n \to \cZ$ and randomized algorithm $f : \cZ \to \cZ'$, $f \circ \cM$ is $(\eps,\delta)$-differentially private.
\end{fact}

This means that any computation based solely on the output of a differentially private function does not affect the privacy guarantee. Refer to Prop. 2.1 in the text by Dwork and Roth \cite{DR14} for a proof. The second property is closure under composition.

\begin{fact}
\label{fact:comp}
	For $(\eps_1, \delta_1)$-differentially private $\cM_1$ and $(\eps_2, \delta_2)$-differentially private $\cM_2$, $\cM_3$ defined by $\cM_3(\vec{x}) = (\cM_1(\vec{x}), \cM_2(\vec{x}))$ is $(\eps_1 + \eps_2, \delta_1 + \delta_2)$-differentially private.
\end{fact}

\begin{fact}
\label{fact:adv_comp}
	For $(\eps, \delta)$-differentially private algorithms $\cM_1,\dots,M_d$, the algorithm $\hat{\cM}_3$ defined by $\hat{\cM}_3(\vec{x}) = (\cM_1(\vec{x}), \dots, \cM_d(\vec{x}))$ is $\left( \eps (e^{\eps}-1)\cdot d + \eps \cdot \sqrt{2d \log \tfrac{1}{d\delta}}, 2d\delta \right)$-differentially private.
\end{fact}

Refer to Theorems 3.14 and 3.20 in \cite{DR14} for proofs.

%%%%
\subsection{Local Model}
In an extreme case, no user trusts any other party with protecting their data; here, we model the dataset as a distributed object where each of $n$ users holds a single row. Each user $i$ provides their data point as input to a randomizing function $\cR$ and publishes the outputs for some analyzer to compute on.

\begin{definition}[Local Model \cite{Warner65,EGS03}]
A protocol $\cP$ in the \emph{local model} consists of two randomized algorithms:
\begin{itemize}
    \item A randomizer $\cR:\cX\to\cY$ mapping data to a message.
    \item An analyzer $\cA:\cY^n\to\cZ$ that computes on a vector of messages.
\end{itemize}
We define its execution on input $\vec{x}\in\cX^n$ as $$\cP(\vec{x}) := \cA(\cR(x_1),\dots,\cR(x_n)).$$ We assume that $\cR$ and $\cA$ have access to an arbitrary amount of public randomness.
\end{definition}

\begin{definition}[Local Differential Privacy \cite{DMNS06,KLNRS08}]
A local protocol $\cP=(\cR,\cA)$ is $(\eps,\delta)$-\emph{differentially private} if $\cR$ is $(\eps,\delta)$-differentially private. The privacy guarantee is over the internal randomness of the users' randomizers and not the public randomness of the protocol.
\end{definition}

For brevity, we typically call these protocols ``locally private.''

%%%%%
\subsection{Shuffle Model}
We focus on differentially private protocols in the shuffle model, which we define below.

\begin{definition}[Shuffle Model \cite{BEMMR+17, CSUZZ19}]
A protocol $\cP$ in the \emph{shuffle model} consists of three randomized algorithms:
\begin{itemize}
\item
    A \emph{randomizer} $\cR: \cX \rightarrow \cY^*$ mapping a datum to a vector of messages.
\item
    A \emph{shuffler} $\cS: \cY^* \rightarrow \cY^*$ that applies a uniformly random permutation to the messages in its input. 
\item
    An \emph{analyzer} $\cA: \cY^* \rightarrow \cZ$ that computes on a permutation of messages.
\end{itemize}
As $\cS$ is the same in every protocol, we identify each shuffle protocol by $\cP = (\cR, \cA)$. We define its execution by $n$ users on input $\vec{x}\in\cX^n$ as
$$
\cP(\vec{x}) := \cA(\cS(R(x_1), \ldots, R(x_n))).
$$
Importantly, we allow $\cR$ and $\cA$ to have parameters that depend on $n$. 
\end{definition}

The following is a definition of differential privacy in this model.
\begin{definition} [Shuffle Differential Privacy \cite{CSUZZ19}]
\label{def:shuffle_dp}
	A protocol $\cP = (\cR, \cA)$ is \emph{$(\eps, \delta)$-shuffle differentially private for $n$ users} if the algorithm $(\cS \circ \cR^n)(\vec{x}) := \cS(\cR(x_1), \ldots, \cR(x_n))$ is $(\eps, \delta)$-differentially private. The privacy guarantee is over the internal randomness of the users' randomizers and not the public randomness of the shuffle protocol.
\end{definition}

For brevity, we typically call these protocols ``shuffle private.''

%%%%%
\subsection{Notation for Histogram and Top-$t$ Selection Problems}
We assume each user $i$ has some private value belonging to the finite set $[d]$ but encodes them as ``one-hot'' binary strings. That is, for any $j\in[d]$, let $e_{j,d}$ be the binary string of length $d$ with zeroes in all entries except for coordinate $j$; user $i$ has data $x_i = e_{j,d}$ for some $j$. Let $\cX_d$ denote the set $\{e_{1,d}, \dots, e_{d,d}\}$ and let $0^d$ denote the binary string of all zeroes.

For any $j\in[d]$, let $\hist_j(\vec{x})$ be the function that takes the vector of one-hot values $\vec{x} \in \{e_{1,d},\dots,e_{d,d}\}^n$ and reports $\tfrac1n\sum_{i=1}^n x_{i,j}$, which is the frequency of $e_{j,d}$ in $\vec{x}$. Let $\hist(\vec{x})$ be shorthand for the vector $(\hist_1(\vec{x}), \dots, \hist_d(\vec{x}))$.

We will use $\ell_\infty$ error to quantify how well a vector $\vec{z}\in\R^d$ estimates the histogram $\hist(\vec{x})$. Specifically, $\norm{\vec{z}-\hist(\vec{x})}_\infty := \max_j |z_j - \hist_j(\vec{x})|$.

\medskip

Having defined histograms, we move on to defining the top-$t$ items. For any vector $\vec{h}\in \R^d$ and value $j\in[d]$, let $\rank_j(\vec{h})$ be the relative magnitude of $h_j$: the index of $h_j$ after sorting $\vec{h}$ in descending order. For any $t\in[d]$, let $\mathrm{top}_t(\vec{h})$ denote the set of $j$ such that $\rank_j(\vec{h}) \leq t$.

We now establish notation to quantify how well a set approximates the top-$t$ items. Let $\hist_{[t]}(\vec{x})$ denote the frequency of the $t$-th largest item: the quantity $\hist_j(\vec{x})$ where $\rank_j(\hist(\vec{x})) = t$.
.
\begin{defn}
\label{def:heavy-hitter}
For any $\vec{x}\in \cX^n_d$, a set of candidates $C \subset [d]$ $\alpha$-approximates the top-$t$ items in $\vec{x}$ if $|C|=t$ and $\hist_j(\vec{x}) > \hist_{[t]}(\vec{x}) - \alpha$ for all $j\in C$.
\end{defn}

Other metrics include precision $p$ (the fraction of items in candidate set $C$ that are actually in the top $t$) and recall $r$ (the fraction of items in the top $t$ that are in $C$). Note that when $|C|=t$, $p=r$ so that the F1 score---the quantity $2\cdot \tfrac{p\cdot r}{p+r}$---is exactly $p=r$.

%% file: protocol.tex
\section{Our Histogram Protocol}
\label{sec:protocol}

A user who executes our protocol's local randomizer $\cR_\flip$ (Algorithm \ref{alg:r-flip}) reports $k+1$ messages. They make their first message by running $\cR_{d,q}$ (Algorithm \ref{alg:rr-onehot}) on their one-hot string. An instance of randomized response, $\cR_{d,q}$ flips each bit of with probability $q$. The user makes the $k$ other messages by running $\cR_{d,q}$ $k$ times on the string $0^d$, with fresh randomness in each execution. This effectively inserts $k$ fake users into the protocol. We will show that the messages from these fake users are sufficiently noisy for differential privacy.

Stacking the $nk+n$ messages results in a $(nk+n) \times d$ binary matrix; to estimate the frequency of $j$, our analyzer $\cA_\flip$ (Algorithm \ref{alg:a-flip}) simply de-biases and re-scales the sum of the $j$-th column.

\begin{algorithm}
\caption{$\cR_\flip$, a randomizer for histograms}
\label{alg:r-flip}

\KwIn{$x \in \cX_d$; implicit parameters $d,k,q$}
\KwOut{$\vec{y} \in (\zo^{d})^{k+1}$}

Initialize $\vec{y}$ as an empty message vector.

Append message generated by $\cR_{d,q}(x)$ to $\vec{y}$

\For{$j \in [k]$}{
    Append message generated by $\cR_{d,q}(0^d)$ to $\vec{y}$
}

\Return{$\vec{y}$}
\end{algorithm}

\begin{algorithm}
\caption{$\cR_{d,q}$, applies randomized response to a binary string}
\label{alg:rr-onehot}

\KwIn{$x \in \zo^d$}
\KwOut{$y \in \zo^d$}
\For{$j \in [d]$}{
    $\mathit{flip}_j \sim \Ber(q)$
    
    \If{$\mathit{flip}_j=1$}{
        $y_j \gets 1-x_j$
    }
    \Else{
        $y_j \gets x_j$
    }
}

\Return{$y$}
\end{algorithm}

\begin{algorithm}
\caption{$\cA_\flip$, an analyzer for histograms}
\label{alg:a-flip}

\KwIn{$\vec{y} \in (\zo^d)^{nk+n}$; implicit parameters $d,k,q$}
\KwOut{$\vec{z} \in \R^d$}

\For{$j\in[d]$}{
    $z_j \gets \tfrac1n \sum_{i=1}^{nk+n} \frac{1}{1-2q} \cdot ( y_{i,j} - q)$
}

\Return{$\vec{z}\gets(z_1,\dots,z_d)$}
\end{algorithm}

Our analysis of the protocol will be built upon two technical claims. The first gives a bound on the size of any confidence interval in terms of parameters $q,k$.

\begin{clm}
\label{clm:acc-per-bin}
Fix any $n\in \N$ and $\beta \in(0,1)$. If $\frac{1}{nk+n} \ln \frac{2}{\beta} \leq q < 1/2$, then the protocol $\cP_\flip = (\cR_\flip, \cA_\flip)$ reports approximate histograms with error behaving as follows:
$$
\forall \vec{x} \in \cX^{n}_d,~ j\in[d]~ \pr{}{\left| z_j - \hist_j(\vec{x}) \right| > 2\sqrt{\frac{k+1}{n} q (1-q) \ln \frac{2}{\beta}} \cdot \left(\frac{1}{1-2q}\right)} \leq \beta
$$
\end{clm}

We will prove this claim in Section \ref{sec:accuracy}. The second claim is a sufficient condition on $q,k$ for $(\eps, \delta)$ shuffle privacy.

\begin{clm}
\label{clm:privacy}
Fix any $\eps > 0$, $\delta < 1/100$, and $n \in \N$. If parameters $q < 1/2$ and $k\in \N$ are chosen such that $q(1-q) \geq \tfrac{33}{5nk} (\tfrac{e^\eps+1}{e^\eps-1})^2 \ln \tfrac{4}{\delta} $, then $\cP_\flip = (\cR_\flip, \cA_\flip)$ is $(\eps, \delta)$-shuffle private.
\end{clm}

We will prove this claim in Section \ref{sec:privacy}. Combining the two claims yields the following confidence interval for the error of any single frequency estimate.

%We will first derive a bound on $k$ to ensure that our protocol guarantees both differential privacy and, for any single $j$, an estimator $z_j$ that strongly concentrates around $\hist_j(\vec{x})$ (Theorem \ref{thm:per-bin}). Then, using a union bound, we show that the maximum error of the entire approximate histogram is low (Theorem \ref{thm:maximum}). This is enough to prove Parts \textit{i,ii} of Theorem \ref{thm:informal}. We prove robustness to manipulation (Part \textit{iii}) in Section \ref{sec:robustness}.

%We remark that the description and analysis of the protocol uses the message universe $\zo^d$ but any such string is equivalent to a subset of $[d]$; part \textit{iii} of Theorem \ref{thm:informal} follows from this observation and our choice of $q$.

\begin{thm}
\label{thm:per-bin}
Fix any $\eps > 0$, $\delta < 1/100$, and $n\in \N$. For any choice of parameter $k > \frac{132}{5n} (\tfrac{e^\eps+1}{e^\eps-1})^2 \ln \tfrac{4}{\delta}$, there is a choice of parameter $q <1/2$ such that the protocol $\cP_\flip=(\cR_\flip,\cA_\flip)$ has the following properties
\begin{enumerate}[a.]
    \item $\cP_\flip$ is $(\eps,\delta)$-shuffle private for inputs from $\cX_d$.
    \item For any $j\in [d]$ and $\vec{x}\in \cX^n_d$, $\cP_\flip(\vec{x})$ reports frequency estimate $z_j$ such that $$|z_j -\hist_j(\vec{x})| < \frac{1}{n}\cdot \frac{e^\eps+1}{e^\eps-1} \cdot \sqrt{\frac{264}{5} \ln \frac{4}{\delta} \ln 20} \cdot g(k) $$ with probability $9/10$, where $g(k)$ monotonically approaches 1 from above. Refer to Figure \ref{fig:per-bin}. %an unbiased estimate of $\hist_j(\vec{x})$. The standard deviation of error is $\leq \tfrac{1}{n}\cdot \tfrac{e^\eps+1}{e^\eps-1} \cdot \sqrt{\tfrac{66}{5} \ln \tfrac{4}{\delta}} \cdot f(k)$ where $f(k)$ is monotone and $\lim_{k\to \infty} f(k) = 1$.
\end{enumerate}
\end{thm}

\begin{figure}
    \centering
    \includegraphics[width=0.45\textwidth]{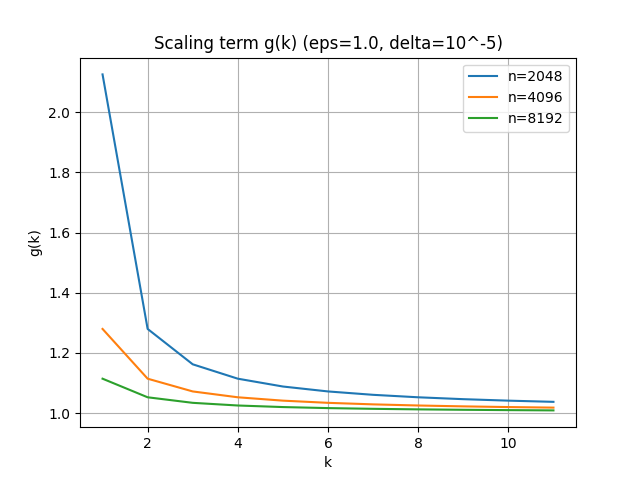}
    \includegraphics[width=0.45\textwidth]{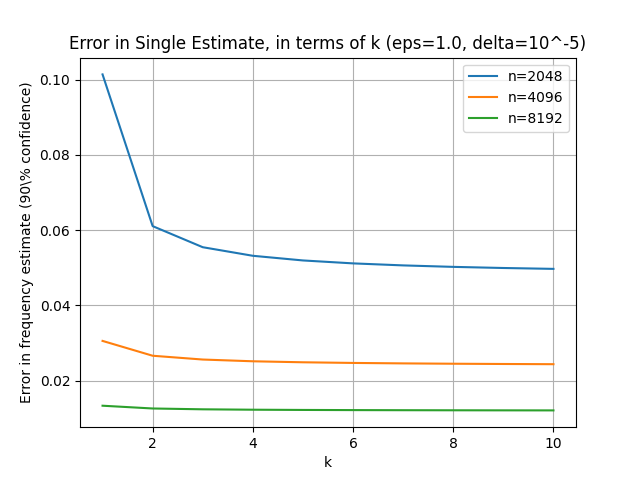}
    \caption{Effect of bandwidth parameter $k$ on scaling term $g(k)$ and error of a single estimate.}
    \label{fig:per-bin}
\end{figure}

\begin{proof}
Because $k$ is sufficiently large, there is a solution $\hat{q}$ to the quadratic equation $q(1-q) = \tfrac{33}{5nk} (\tfrac{e^\eps+1}{e^\eps-1})^2 \ln \tfrac{4}{\delta} $ that lies in the interval $(0,1/2)$. Also, let $\tilde{q} \gets \frac{1}{nk+n} \ln \frac{2}{\beta} < 1/2$.

When we set $q\gets \max(\hat{q},\tilde{q})$, Part \textit{a} follows immediately from Claim \ref{clm:privacy} and the error $|z_j - \hist_j(\vec{x})|$ is at most $$\max\left(\frac{1}{n}\cdot \frac{e^\eps+1}{e^\eps-1} \cdot \sqrt{\frac{264}{5} \ln \frac{4}{\delta} \ln 20 } ,~ \frac{2}{n}\cdot \ln 20 \right) \cdot \frac{1}{1-2q}$$ with  probability $9/10$ via Claim \ref{clm:acc-per-bin} and the bound $(k+1)/k\leq 2$. Note that both forms of $q$ approach zero as $k$ increases, so $1/(1-2q)$ is a monotonically decreasing function of $k$ as desired. Finally, the term $\tfrac{2}{n}\cdot \ln 20$ must be the smaller of the two due to our bound on $\delta$.
\end{proof}

\medskip

We now iterate on our analysis to derive a bound on the maximum error. Parts \textit{i} and \textit{ii} in Theorem \ref{thm:informal} are immediate corollaries.
\begin{thm}
\label{thm:maximum}
Fix any $\eps > 0$, $\delta < 1/100$, and $n\in \N$. For any choice of parameter $k > \max \left( \tfrac{132}{5n} (\tfrac{e^\eps+1}{e^\eps-1})^2 \ln \tfrac{4}{\delta}, \tfrac{2}{n}\ln20d -1 \right)$, there is a choice of parameter $q <1/2$ such that the protocol $\cP_\flip=(\cR_\flip,\cA_\flip)$ has the following properties
\begin{enumerate}[a.]
    \item $\cP_\flip$ is $(\eps,\delta)$-shuffle private for inputs from $\cX_d$.
    \item For any $\vec{x}\in \cX^n_d$, $\cP_\flip(\vec{x})$ reports approximate histogram $\vec{z}$ such that the maximum error is $$\norm{\vec{z} - \hist(\vec{x})}_\infty < \max\left(\frac{1}{n}\cdot \frac{e^\eps+1}{e^\eps-1} \cdot \sqrt{\frac{264}{5} \ln \frac{4}{\delta} \ln 20d} ,~ \frac{2}{n}\cdot \ln 20d \right) \cdot f(k)$$ with probability $9/10$, where $f(k)$ monotonically approaches 1 from above. Refer to Figure \ref{fig:max}.  
\end{enumerate}
\end{thm}

\begin{figure}
    \centering
    \includegraphics[width=0.45\textwidth]{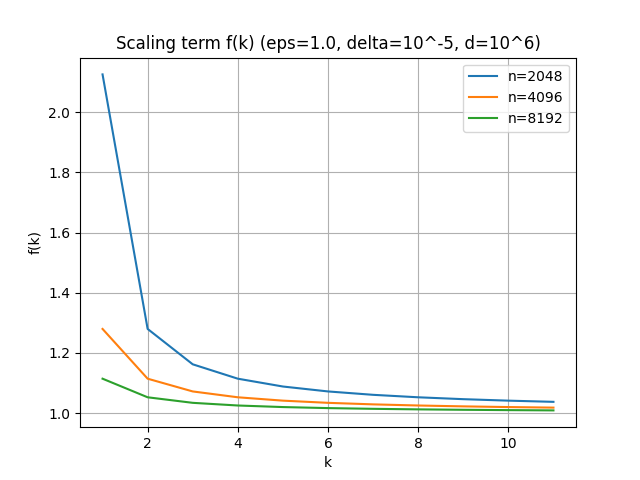}
    \includegraphics[width=0.45\textwidth]{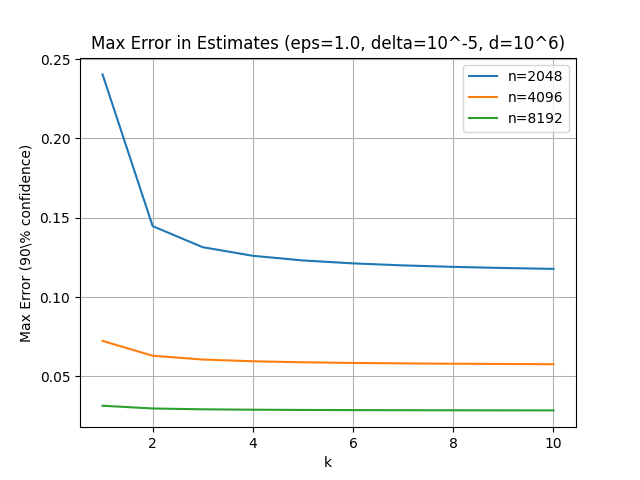}
    \caption{Effect of bandwidth parameter $k$ on scaling term $f(k)$ and maximum error.}
    \label{fig:max}
\end{figure}

\begin{proof}
The following is immediate from setting $\beta = 1/10d$  in Claim \ref{clm:acc-per-bin} and a union bound:
\begin{coro}
\label{coro:acc-max}
Fix any $n\in \N$. If $\frac{1}{nk+n} \ln 20d \leq q < 1/2$, then $\cP_\flip$ reports a histogram with maximum error behaving as follows:
$$
\forall \vec{x} \in \cX^{n}_d~ \pr{}{\norm{\cP_\flip(\vec{x}) - \hist(\vec{x})}_\infty > 2 \sqrt{\frac{k+1}{n} q(1-q) \ln 20d } \cdot \left(\frac{1}{1-2q}\right)} \leq 1/10
$$
\end{coro}

Because $k$ is sufficiently large, there is a solution $\hat{q}$ to $q(1-q) = \tfrac{33}{5nk} (\tfrac{e^\eps+1}{e^\eps-1})^2 \ln \tfrac{4}{\delta} $ such that $\hat{q} \in (0,1/2)$. Similar to before, we will choose $q$ to be the maximum of $\hat{q}$ and $\tilde{q} \gets \frac{1}{nk+n} \ln 20d$. The theorem follows from Claim \ref{clm:privacy} and Corollary \ref{coro:acc-max} (where $1/1-2q$ is again the desired monotonically decreasing function of $k$)
\end{proof}

We will prove Claims \ref{clm:acc-per-bin} and \ref{clm:privacy} in the following two subsections. In Subsection \ref{sec:robustness}, we will bound the impact of corrupt users (Part \textit{iii} of Theorem \ref{thm:informal}).

\input{accuracy}

\input{privacy}

\input{robustness}

\input{heavyhitters}

%% file: accuracy.tex
\subsection{Accuracy of $\cP_\flip$}
\label{sec:accuracy}
In this section, we show how to obtain confidence intervals of the per-bin error of $\cP_\flip$ (Claim \ref{clm:acc-per-bin}). To prove the claim, we will first analyze the bias and variance of each estimate.

\begin{clm}
\label{clm:bias-and-variance}
Fix any $q < 1/2$ and $n,k\in \N$. For any index $j \in [d]$ and data $\vec{x}\in\cX^n_d$, the protocol $\cP_\flip = (\cR_\flip, \cA_\flip)$ reports an unbiased estimate of $\hist_j(\vec{x})$ with variance $\frac{k+1}{n} \cdot q(1-q)\cdot (\tfrac{1}{1-2q})^2$.
\end{clm}

\begin{proof}
For $n < i\leq nk+n$, we take $x_i=0^d$. This will correspond to the empty data vector of a fabricated user. Recall that $z_j = \tfrac1n \sum_{i=1}^{nk+n} \frac{1}{1-2q} \cdot ( y_{i,j} - q)$ is the protocol's estimate of $\hist_j(\vec{x}) = \tfrac1n\sum_{i=1}^n x_{i,j}$. We will first show each term $\frac{1}{1-2q} \cdot ( y_{i,j} - q)$ is an unbiased estimate of the bit $x_{i,j}$.\footnote{Technically, shuffling means that the $i$-th message does not correspond to the $i$-th user. But summation is symmetric so we simply avoid inserting permutation notation for neatness.}
\begin{align*}
\ex{}{\frac{1}{1-2q} \cdot ( y_{i,j} - q)} &= \frac{1}{1-2q} \cdot (\ex{}{y_{i,j}} - q)\\
    &= \frac{1}{1-2q} \cdot ( ( (1-x_{i,j})\cdot q + x_{i,j}\cdot (1-q) ) -q ) \tag{see $\cR_{d,q}$}\\
    &= \frac{1}{1-2q} \cdot ( q-x_{i,j}q + x_{i,j}- x_{i,j}q - q)\\
    &= x_{i,j}
\end{align*}

Next, we derive the variance of the term:
\begin{align*}
\var{}{\frac{1}{1-2q} \cdot ( y_{i,j} - q)} &= \left(\frac{1}{1-2q}\right)^2 \cdot \var{}{y_{i,j}} \\
    &= \left(\frac{1}{1-2q}\right)^2 \cdot q(1-q)
\end{align*}
The second equality comes from the fact that $y_{i,j}$ is drawn from either $\Ber(q)$ or $\Ber(1-q)$, which have the same variance.

Because $z_j$ is the summation over $nk+n$ terms (normalized by $\tfrac1n$), the variance of $z_j$ is $\frac{k+1}{n}\cdot q(1-q) \cdot (\tfrac{1}{1-2q})^2$ by independence. Finally, $\ex{}{z_j}=\hist_j(\vec{x})$ due to linearity of expectation and the fact that we normalize by $n$.
\end{proof}

To arrive at Claim \ref{clm:acc-per-bin}, we will show that the protocol's estimates are sums of bounded random variables. This will allow us to deploy a concentration inequality.
\begin{proof}[Proof of Claim \ref{clm:acc-per-bin}]
We first expand the random variable in question as
\begin{equation}
\label{eq:confidence-1}
\frac{1}{n} \sum_{i=1}^{nk+n} \frac{1}{1-2q} \cdot ( y_{i,j} - q) - x_{i,j}
\end{equation}
where we again use $x_{i,j}$ (resp. $y_{i,j}$) to denote the $j$-th bit in the data (resp. message) sent by user $i$. When $i > n$, $i$ corresponds to the index of a fabricated user; in this case, $x_{i,j}=0$.

In the proof of Claim \ref{clm:bias-and-variance}, we saw that each term in $\sum_{i=1}^{nk+n} \frac{1}{1-2q} \cdot ( y_{i,j} - q)$ is an independent random variable with mean $x_{i,j}$ and variance $\frac{1}{(1-2q)^2} \cdot q \cdot (1-q)$. Naturally, this means each term in \eqref{eq:confidence-1} is an independent random variable with mean zero and variance $\frac{1}{(1-2q)^2} \cdot q \cdot (1-q)$.

We now add the observation that each term in \eqref{eq:confidence-1} has maximum magnitude $m= \frac{1-q}{1-2q}$. This follows from the fact that $x_{i,j},y_{i,j}\in \zo$ and $q < 1/2$. We show that the variance of the summation is at least $m^2\ln \frac{2}{\beta}$:
\begin{align*}
\var{}{\sum_{i=1}^{nk+n} \frac{1}{1-2q} \cdot ( y_{i,j} - q)} &= \frac{nk+n}{(1-2q)^2} \cdot q\cdot (1-q) \\
    &= \left(\frac{1-q}{1-2q}\right)^2 \cdot \frac{q}{1-q} \cdot (nk+n) \\
    &\geq \left(\frac{1-q}{1-2q}\right)^2 \ln\frac{2}{\beta} \\
    &= m^2 \ln\frac{2}{\beta}
\end{align*}

Because we have lower bounded the variance of the sum of bounded independent variables, the claim follows from an additive Chernoff bound.
\end{proof}

%% file: privacy.tex
\subsection{Privacy of $\cP_\flip$}
\label{sec:privacy}
In this section, we derive the range of $q,k$ for which $(\eps,\delta)$-privacy will hold (Claim \ref{clm:privacy}). The proof will proceed as follows: design a series of algorithms $\cM_1, \cM_2, \dots $ such that $\cP_\flip$ is private whenever $\cM_1$ is private, $\cM_1$ is private whenever $\cM_2$ is private, and so on. Then we study the privacy of the final algorithm.

%%%%
\subsubsection{Step One}
We first consider $\cC_{m,d,q}$ (Algorithm \ref{alg:sim-protocol}). It takes one user's data as input, constructs $m$ copies of $0^d$ and executes the randomization algorithm $\cR_{d,q}$ on all $m+1$ strings. When $m=nk$, this algorithm simulates the set of messages produced by any single user and the fabricated users in our protocol $\cP_\flip$.

\begin{algorithm}
\caption{$\cC_{m,d,q}$}
\label{alg:sim-protocol}

\KwIn{$x \in \{0^d \} \cup \cX_d$}
\KwOut{$\vec{y} \in (\zo^d)^{m+1}$}

Construct $\vec{x} \gets (x, \underbrace{0^d, \dots, 0^d}_{m ~\textrm{copies}})$

\Return{$\vec{y} \gets (\cS\circ \cR^{m+1}_{d,q})(\vec{x})$}

\end{algorithm}

We claim that privacy of our protocol follows from privacy of this new algorithm.

\begin{clm}
If $\cC_{nk,d,q}$ is $(\eps,\delta)$-differentially private for inputs from $\cX_d$, then $\cP_\flip = (\cR_\flip,\cA_\flip)$ is $(\eps,\delta)$-shuffle private for inputs from $\cX_d$.
\end{clm}
\begin{proof}
In an execution of $(\cS \circ \cR^n_\flip)(\vec{x})$, $\cR_{d,q}$ gets run on the values $x_1, \dots, x_n, x_{n+1} , \dots, x_{nk+n}$ --- where $x_{n+1} = 0^d, \dots, x_{nk+n} = 0^d$ --- and all $nk+n$ messages are shuffled together. For any user $i$, we can decompose it into two stages: (1) run $\cR_{d,q}$ on the values $x_i, x_{n+1} = 0^d, \dots, x_{nk+n} = 0^d$ and shuffle the output then (2) run $\cR_{d,q}$ on the values $x_1,\dots,x_{i-1}, x_{i+1},\dots, x_n$ and shuffle all $nk+n$ messages. The first stage is precisely $\cC_{nk,d,q}$ and the second is a post-processing of its output. Thus, privacy follows from post-processing (Fact \ref{fact:post}).
\end{proof}

%%%%
\subsubsection{Step Two}
In this step, we argue that we only need to concern ourselves with the $d=2$ case. Consider any $j,j'\in[d]$ where $j<j'$. Changing user data from $e_{j,d}$ to $e_{j',d}$ only affects the one-hot strings in positions $j$ and $j'$. Because $\cC_{m,d,q}$ operates by performing independent bit-flipping on the one-hot strings (via $\cR_{d,q}$), it can essentially be decomposed into two phases: bit-flip positions $j,j'$ (which depend on the user data) and then bit-flip on the rest of the bits (a post-processing that is independent of the user's data). We make this decomposition explicit in Algorithm \ref{alg:helper}.

%These phases correspond to the next pair of algorithms in our sequence, $\cC_{m,2,q}$ and $\cC_{m,d,q,j,j'}$ (Algorithm \ref{alg:helper}) where $j,j'\in [d]$ and $j\neq j'$.

\begin{algorithm}
\caption{$\cC_{m,d,q,j,j'}$}
\label{alg:helper}

\KwIn{$x \in \{e_{j,d}, e_{j',d}\}$}
\KwOut{$\vec{y} \in (\zo^d)^*$}

\If{$x = e_{j,d}$}{
    $u \gets 10$
}
\Else{
    $u \gets 01$
}
%Construct $\vec{u} \gets (u, \underbrace{00, \dots, 00}_{m ~\textrm{copies}})$

$\vec{v} \gets \cC_{m,2,q}(u)$

$\vec{w} \gets \cC_{m,d-2,q}(0^{d-2})$

$\vec{y} \gets$ empty list

\For{$i \in [m+1]$}{
    $y_i \gets (\dots, w_{i,j-1},v_{i,1},w_{i,j+1}, \dots, w_{i,j'-1},v_{i,2},w_{i,j'+1}, \dots)$
    
    Append $y_i$ to $\vec{y}$
}
\Return{$\vec{y}$}

\end{algorithm}

\begin{clm}
If, for every $j,j'\in [d]$ where $j<j'$, $\cC_{m,d,q,j,j'}$ is $(\eps,\delta)$-differentially private for inputs from $\{e_{j,d}, e_{j',d}\}$, then $\cC_{m,d,q}$ is $(\eps,\delta)$-differentially private for inputs from $\cX_d$.
\end{clm}
\begin{proof}
To prove $\cC_{m,d,q}$ is $(\eps,\delta)$-differentially private for $\cX_d$, it suffices to show the inequalities below are true for every $j,j',Y$:
\begin{align}
    \pr{}{\cC_{m,d,q}(e_{j,d}) \in Y} &\leq e^\eps\cdot \pr{}{\cC_{m,d,q}(e_{j',d}) \in Y} +\delta \label{eq:priv1}\\
    \pr{}{\cC_{m,d,q}(e_{j',d}) \in Y} &\leq e^\eps\cdot \pr{}{\cC_{m,d,q}(e_{j,d}) \in Y} +\delta \label{eq:priv2}
\end{align}

If $\cC_{m,d,q,j,j'}$ is $(\eps,\delta)$-differentially private, we have that
\begin{align}
    \pr{}{\cC_{m,d,q,j,j'}(e_{j,d}) \in Y} &\leq e^\eps\cdot \pr{}{\cC_{m,d,q,j,j'}(e_{j',d}) \in Y} +\delta\\
    \pr{}{\cC_{m,d,q,j,j'}(e_{j',d}) \in Y} &\leq e^\eps\cdot \pr{}{\cC_{m,d,q,j,j'}(e_{j,d}) \in Y} +\delta 
\end{align}
In the remainder of the proof, we will argue that $\cC_{m,d,q,j,j'}(e_{j,d})$ has the same distribution as $\cC_{m,d,q}(e_{j,d})$; a completely symmetric argument holds for the equivalence between $\cC_{m,d,q,j,j'}(e_{j',d})$ and $\cC_{m,d,q}(e_{j',d})$. Inequalities \eqref{eq:priv1} and \eqref{eq:priv2} will therefore hold by substitution.

Pick any $k\notin \{j,j'\}$ and any message index $i \in [m+1]$. Notice that when we obtain $\vec{y}$ from $\cC_{m,d,q}(e_{j,d})$, $y_{i,k}$ is an independent bit that has value 1 with probability $q$ (since bit $k$ has to flip from 0 to 1). But this is exactly the same distribution as in $\cC_{m,d,q,j,j'}(e_{j,d})$.

Now consider $(y_{1,j},\dots,y_{m+1,j}, y_{1,j'},\dots,y_{m+1,j'})$ when obtained from $\cC_{m,d,q}(e_{j,d})$. By construction, we know that there exists one uniformly random index $i$ such that $y_{i,j}\sim \Ber(1-q), y_{i,j'}\sim \Ber(q)$ and, for every other index $\hat{i}$, $y_{\hat{i},j}$ and $y_{\hat{i},j'}$ are independent samples from $\Ber(q)$. But again this is the same as $\cC_{m,d,q,j,j'}(e_{j,d})$.
\end{proof}

\begin{clm}
If $\cC_{m,2,q}$ is $(\eps,\delta)$-differentially private for inputs from $\cX_2$, then $\cC_{m,d,q,j,j'}$ is $(\eps,\delta)$-differentially private for inputs from $\{e_{j,d}, e_{j',d}\}$.
\end{clm}
\begin{proof}
The claim is immediate from the fact that $\cC_{m,d,q,j,j'}$ is executing $\cC_{m,2,q}$ on a value that is obtained from the user input and then post-processing the algorithm's output.
\end{proof}

%%%%
\subsubsection{Step Three}
In this section, we reduce the privacy of $\cC_{m,2,q}$ to that of $\cB_{m,q}$ (Algorithm \ref{alg:b}). This algorithm generates a vector of four randomized integers via $\cM(m,q)$ (Algorithm \ref{alg:n}). Then it computes a binary string $j \gets \cR_{2,q}(x)$ where $x$ is the input to $\cB_{m,q}$. Finally it increments the integer at the position encoded by $j$.

We design $\cM(m,q)$ to generate the histogram of the messages produced by $m$ fabricated users. This means $\cB_{m,q}(x)$ is sufficient to simulate $\cC_{m,2,q}(x)$. In turn, it suffices to prove that $\cB_{m,q}$ is private. Then we argue that, whenever $m,q$ lie in a particular range, the noise produced by $\cM(m,q)$ is enough to ensure $\cB_{m,q}$ satisfies $(\eps,\delta)$-differential privacy.

\begin{algorithm}
\caption{$\cM$, a multinomial noise generator}
\label{alg:n}
\KwIn{$m \in \N,q \in(0,1)$}
\KwOut{$\vec{f} \in \Z_{\geq 0}^4$}

Initialize $\vec{f} = (f_1,f_2,f_3,f_4) \gets (0,0,0,0)$

\For{$i \in [m]$}{
    Sample $j \sim \cR_{2,q}(00)$
    
    $f_{j+1} \gets f_{j+1} + 1$ \tcc{binary string $j$ maps to a number between 0 and 3}
}

\Return{$\vec{f}$}
\end{algorithm}

\begin{algorithm}
\caption{$\cB_{m,q}$}
\label{alg:b}

\KwIn{$x \in \zo^2$}
\KwOut{$\vec{y} \in \Z_{\geq 0}^4$}

Initialize $\vec{y} = (y_1,y_2,y_3,y_4)$ with noise from $\cM(m,q) $

Sample $j \sim \cR_{2,q}(x)$

$y_{j+1} \gets y_{j+1} + 1$ \tcc{binary string $j$ maps to a number between 0 and 3}

\Return{$\vec{y}$}
\end{algorithm}

\begin{clm}
If $\cB_{m,q}$ is $(\eps,\delta)$-differentially private for inputs from $\cX_2$, then $\cC_{m,2,q}$ is $(\eps,\delta)$-differentially private for inputs from $\cX_2$.
\end{clm}

\begin{proof}
Consider the post-processing algorithm which takes $\vec{y}$ produced by $\cB_{m,q}(x)$ and generates a uniformly random vector $\vec{w} \in \cX^{m+1}_2$ such that $y_j$ describes the frequency of the binary string corresponding to $j$ in $\vec{w}$. This is exactly the distribution of $\cC_{m,2,q}$ so privacy follows from post-processing.
\end{proof}

\begin{clm}
\label{clm:bmq}
Fix any $\eps > 0$ and $\delta < 1/100$. If $q<1/2$ and $mq(1-q) \geq \tfrac{33}{5} \left(\tfrac{e^\eps+1}{e^\eps-1}\right)^2 \ln(4/\delta)$, then $\cB_{m,q}$ is $(\eps,\delta)$-differentially private for inputs from $\cX_2$.
\end{clm}

Our proof makes formal the following steps. Recall that the algorithm encodes the user's value $x$ via the randomized algorithm $\cR_{2,q}$. Changing $x$ from 01 to 10 will affect the probability mass function (PMF) of this encoding, but we note that the PMF only changes at two elements of the support, 01 and 10 (see Table \ref{tab:pmf}). This means we need only focus on how the noise produced by $\cM(m,q)$ behaves on those elements.

We essentially argue that a noise vector $\vec{f}=(f_1,\dots,f_4)$ produced by $\cM(m,q)$ has properties that are in line with the binomial and Gaussian distributions: we show that a sample from $\cM(m,q)$ is very likely to be in a set $F$ and any outcome in $F$ has the property that its probability is within $e^\eps$ of a ``neighboring'' outcome's probability. We formalize this in the two claims below, proven in Appendix \ref{apdx:technical}.

\begin{clm}
\label{clm:noise-concentration}
   Fix $m\in \N$ and $q, \delta \in (0,1)$. Define
    \begin{align*}
    \Delta &:= \sqrt{3mq(1-q)\ln \frac{4}{\delta}} \cdot \frac{q(1-q)}{1-q(1-q)} \\
    U &:= mq(1-q) + \Delta + \sqrt{3(mq(1-q) +\Delta) \ln\frac{4}{\delta}} \\
    L &:= mq(1-q) - \Delta - \sqrt{3(mq(1-q) +\Delta) \ln\frac{4}{\delta}}
    \end{align*}
    Let $F\subset \Z^4$ denote the set of vectors where $\vec{t} \in F$ if and only if $t_2,t_3 \in [L,U]$. If $mq(1-q) > \frac{9}{2}\ln (4/\delta)$, then
    \begin{equation}
    \label{eq:bmq-0}
    \pr{\vec{f} \sim \cM(m,q)}{ \vec{f} \notin F} \leq \delta
    \end{equation}
\end{clm}

\begin{clm}
\label{clm:good-noise}
Fix any $\eps>0$ and $\delta < 1/100$. Define $F$ as in Claim \ref{clm:noise-concentration}. If $mq(1-q) \geq \tfrac{33}{5} \left(\tfrac{e^\eps+1}{e^\eps-1}\right)^2 \ln(4/\delta)$, then for any $\vec{y} = (y_1,\dots,y_4)$,
\begin{align*}
    \pr{\vec{f} \gets \cM(m,q) }{\vec{f} = (y_1,y_2-1,y_3,y_4) ,\vec{f} \in F } &\leq e^\eps \cdot \pr{\vec{f} \gets \cM(m,q) }{\vec{f} = (y_1,y_2,y_3 - 1,y_4) } \\
    \pr{\vec{f} \gets \cM(m,q) }{\vec{f} = (y_1,y_2,y_3 - 1, y_4) ,\vec{f} \in F } &\leq e^\eps \cdot \pr{\vec{f} \gets \cM(m,q) }{\vec{f} = (y_1,y_2 - 1,y_3,y_4) }
\end{align*}
\end{clm}

The rest of this section is dedicated to proving Claim \ref{clm:bmq}.

\begin{proof}[Proof of Claim \ref{clm:bmq}]
For any $Y \subset \Z^4$, we will prove $$\pr{}{\cB_{m,q}(01) \in Y } \leq e^\eps \cdot \pr{}{\cB_{m,q}(10) \in Y} + \delta.$$ The inequality $\pr{}{\cB_{m,q}(10) \in Y } \leq e^\eps \cdot \pr{}{\cB_{m,q}(01) \in Y} + \delta$ will hold by completely symmetric arguments.

%At a high level, the proof uses the fact that the frequency of the message 01 produced by $\cR_{2,q}(00)$ is a binomial distribution, as well as the fact that the frequency of 10 after sampling the frequency of 01 is also a binomial distribution.

We begin by using Claim \ref{clm:noise-concentration} to rewrite $\pr{}{\cB_{m,q}(01) \in Y }$:
\begin{align*}
    &\pr{}{\cB_{m,q}(01) \in Y }\\
    ={}& \sum_{\vec{y}\in Y} \pr{}{\cB_{m,q}(01) = \vec{y}} \\
    ={}& \sum_{\vec{y}\in Y} \sum_{j \in \{0,1,2,3\}} \pr{}{\cR_{2,q}(01) = j} \cdot \pr{\vec{f} \gets \cM(m,q) }{\vec{f}=\vec{y}-e_{j+1,4} } \\
    ={}& \sum_{\vec{y}\in Y} \sum_{j \in \{0,1,2,3\}} \pr{}{\cR_{2,q}(01) = j} \cdot \\
    & \left( \pr{\vec{f} \gets \cM(m,q) }{\vec{f}=\vec{y}-e_{j+1,4} , \vec{f}\in F} + \pr{\vec{f} \gets \cM(m,q) }{\vec{f}=\vec{y}-e_{j+1,4} , \vec{f}\notin F} \right)\\
    \leq{}& \left( \sum_{\vec{y}\in Y} \sum_{j \in \{0,1,2,3\}} \pr{}{\cR_{2,q}(01) = j} \cdot \pr{\vec{f} \gets \cM(m,q) }{\vec{f}=\vec{y}-e_{j+1,4}, \vec{f}\in F } \right) + \delta \stepcounter{equation} \tag{\theequation} \label{eq:bmq-1} 
\end{align*}

\begin{table}
    \centering
    \begin{tabular}{|c||c|c||c|}
    \hline
         & $\cR_{2,q}(10)$ & $\cR_{2,q}(01)$ & $\cR_{2,q}(00)$ \\ \hline
         0=00 & $q(1-q)$ & $q(1-q)$ & $(1-q)^2$ \\ \hline
         1=01 & $q^2$ & $(1-q)^2$ & $q(1-q)$ \\ \hline
         2=10 & $(1-q)^2$ & $q^2$ & $q(1-q)$ \\ \hline
         3=11 & $q(1-q)$ & $q(1-q)$ & $q^2$ \\ \hline
    \end{tabular}
    \caption{Probability mass functions of three distributions $\cR_{2,q}(10), \cR_{2,q}(01), \cR_{2,q}(00)$}
    \label{tab:pmf}
\end{table} 

We will upper bound the inner summation. Expanding out the terms, we have
\begin{align*}
& \sum_{j \in \{0,1,2,3\}} \pr{}{\cR_{2,q}(01) = j} \cdot \pr{\vec{f} \gets \cM(m,q) }{\vec{f}=\vec{y}-e_{j+1,4}, \vec{f}\in F } \\
={}& \pr{}{\cR_{2,q}(01) = 0} \cdot \pr{\vec{f} \gets \cM(m,q) }{\vec{f} = (y_1-1,y_2,y_3,y_4) ,\vec{f} \in F }  \\
&+ \pr{}{\cR_{2,q}(01) = 1} \cdot\pr{\vec{f} \gets \cM(m,q) }{\vec{f} = (y_1,y_2-1,y_3,y_4) ,\vec{f} \in F }  \\
&+ \pr{}{\cR_{2,q}(01) = 2} \cdot \pr{\vec{f} \gets \cM(m,q) }{\vec{f} = (y_1,y_2,y_3-1,y_4) ,\vec{f} \in F }  \\
&+ \pr{}{\cR_{2,q}(01) = 3} \cdot \pr{\vec{f} \gets \cM(m,q) }{\vec{f} = (y_1,y_2,y_3,y_4-1) ,\vec{f} \in F }  \stepcounter{equation} \tag{\theequation} \label{eq:bmq-2}
\end{align*}

As displayed in Table \ref{tab:pmf}, note that $$\pr{}{\cR_{2,q}(01) = 0}=\pr{}{\cR_{2,q}(10) = 0} \textrm{ and } \pr{}{\cR_{2,q}(01) = 3}=\pr{}{\cR_{2,q}(10) = 3}.$$ Also, $$\pr{}{\cR_{2,q}(01) = 1}=\pr{}{\cR_{2,q}(10) = 2} \textrm{ and } \pr{}{\cR_{2,q}(01) = 2}=\pr{}{\cR_{2,q}(10) = 1}.$$ By substitution,
\begin{align*}
\eqref{eq:bmq-2} ={}& \pr{}{\cR_{2,q}(10) = 0} \cdot \pr{\vec{f} \gets \cM(m,q) }{\vec{f} = (y_1-1,y_2,y_3,y_4) ,\vec{f} \in F } \\
&+ \pr{}{\cR_{2,q}(10) = 2} \cdot\pr{\vec{f} \gets \cM(m,q) }{\vec{f} = (y_1,y_2-1,y_3,y_4) ,\vec{f} \in F }  \\
&+ \pr{}{\cR_{2,q}(10) = 1} \cdot \pr{\vec{f} \gets \cM(m,q) }{\vec{f} = (y_1,y_2,y_3-1,y_4) ,\vec{f} \in F }  \\
&+ \pr{}{\cR_{2,q}(10) = 3} \cdot \pr{\vec{f} \gets \cM(m,q) }{\vec{f} = (y_1,y_2,y_3,y_4-1) ,\vec{f} \in F }   \stepcounter{equation} \tag{\theequation} \label{eq:bmq-3}
\end{align*}

By combining \eqref{eq:bmq-1} through \eqref{eq:bmq-3} with Claim \ref{clm:good-noise}, we have
\begin{align*}
    \eqref{eq:bmq-1} \leq{}& e^\eps \cdot \left( \sum_{\vec{y}\in Y} \sum_{j \in \{0,1,2,3\}} \pr{}{\cR_{2,q}(10) = j} \cdot \pr{\vec{f} \gets \cM(m,q) }{\vec{f}=\vec{y}-e_{j+1,4} } \right) + \delta \\
    \leq{}& e^\eps \cdot \left( \sum_{\vec{y}\in Y} \pr{}{\cB_{m,q}(10) = \vec{y}} \right) + \delta \\
    ={}& e^\eps \cdot \pr{}{\cB_{m,q}(10) \in Y} + \delta
\end{align*}
which completes the proof.
\end{proof}

%% file: robustness.tex
\subsection{Robustness of $\cP_\flip$ to Manipulation}
\label{sec:robustness}

As mentioned in the Introduction, \emph{manipulation attacks} have been studied in the local model of privacy \cite{AJL04,MN06,CJG19,CSU19}. We assume there is a coalition of $m$ users are corrupted who send specially crafted messages to skew the output of the protocol. In this section, we adapt this definition to the shuffle model and we upper bound the impact of corrupt users on the estimates produced by $\cP_\flip$. We also show that our protocol's robustness compares favorably with prior work.

A baseline attack against any histogram protocol is to simply run the randomizer on incorrect input. This introduces $m/n$ bias to a single frequency estimate. But we can in fact bound the error of \emph{any} attack against $\cP_\flip$.
\begin{thm}
For any $\eps>0$, $\delta<1/100$, and $n \in \N$, choose $q,k$ as in Theorem \ref{thm:per-bin}. For any input $\vec{x}\in\cX^n_d$, any value $j\in[d]$, and any coalition of $m$ corrupt users $M\subset [n]$, the error of $\cP_\flip$ on $\hist_j$ is $$\frac{1}{n}\cdot \frac{e^\eps+1}{e^\eps-1} \cdot \sqrt{\frac{264}{5} \ln \frac{4}{\delta} \ln 20} \cdot g(k) + \frac{m}{n}\cdot (k+1)\cdot g(k)$$ with $90\%$ probability.
\end{thm}
Note that $k=O(1)$ whenever $n=\Omega(\tfrac{1}{\eps^2} \log\tfrac{1}{\delta})$ and recall $g(k)$ approaches 1. In this regime, the above theorem implies that any attack launched by corrupt users is only a constant factor worse than the baseline attack.

\begin{proof}
Define the function $u(i,n) := ((i-1)\textrm{ mod } n)+1$. For $i\in[n]$, let $y_{i,j}$ be the $j$-th bit of the first message produced by user $i$ (the output of $\cR_{d,q}(x_i)$). For $i > n$, let $y_{i,j}$ be the $j$-th bit of the $\lceil i / n\rceil$-th message produced by user $u(i,n)$ (the output of $\cR_{d,q}(0^d)$).

Recall that the analyzer computes $z_j \gets \tfrac1n \sum_{i=1}^{nk+n} \frac{1}{1-2q} \cdot ( y_{i,j} - q)$. Let $z^{\mathrm{hon}}_j, y^{\mathrm{hon}}_{i,j}$ (resp. $z^{\mathrm{cor}}_j, y^{\mathrm{cor}}_{i,j}$) denote the random variables from an honest (resp. corrupted) execution of the protocol. Theorem \ref{thm:maximum} ensures that $$|z^{\mathrm{hon}}_j-\hist_j(\vec{x})| < \frac{1}{n}\cdot \frac{e^\eps+1}{e^\eps-1} \cdot \sqrt{\frac{264}{5} \ln \frac{4}{\delta} \ln 20} \cdot g(k)$$ with probability $9/10$. Via the triangle inequality, will suffice to show that $$|z^{\mathrm{hon}}_j- z^{\mathrm{cor}}_j| < \frac{m}{n}\cdot (k+1)\cdot g(k)$$ with probability $1$.

By construction, $|y^{\mathrm{hon}}_{i,j} - y^{\mathrm{cor}}_{i,j}| \in \zo$ for any $j\in[d]$ and $i \in M$. Also, for all $i\notin M$, $y^{\mathrm{hon}}_{i,j}$ is identically distributed with $y^{\mathrm{cor}}_{i,j}$. This means
\begin{align*}
|z^{\mathrm{hon}}_j-z^{\mathrm{cor}}_j| &= \left| \frac{1}{n} \sum_{ \{i| u(i,n) \in M\}} \frac{1}{1-2q} \cdot (y^{\mathrm{hon}}_{i,j} - y^{\mathrm{cor}}_{i,j}) \right| \\
    &\leq \frac{1}{n} \cdot \sum_{\{i| u(i,n) \in M\}} \frac{1}{1-2q} \\
    &= \frac{m}{n} \cdot (k+1) \cdot \frac{1}{1-2q}
\end{align*}
This concludes the proof, since $g(k)$ is precisely $1/(1-2q)$ where $q$ depends on $k$.
\end{proof}

Now we bound the maximum error. Because the proof essentially generalizes the prior one, we omit it for brevity.
\begin{thm}
For any $\eps>0$, $\delta<1/100$, and $n \in \N$, choose $q,k$ as in Theorem \ref{thm:maximum}. For any input $\vec{x}\in\cX^n_d$ and any coalition of $m$ corrupt users, the maximum error of $\cP_\flip$ is $$\max\left(\frac{1}{n}\cdot \frac{e^\eps+1}{e^\eps-1} \cdot \sqrt{\frac{264}{5} \ln \frac{4}{\delta} \ln 20d} ,~ \frac{2}{n}\cdot \ln 20d \right) \cdot f(k) + \frac{m}{n}\cdot (k+1)\cdot f(k)$$ with $90\%$ probability.
\end{thm}

We note that the resilience of the protocol stems from an implicit assumption that every user---both honest and corrupt---sends exactly $k+1$ messages to the shuffler. In principle, a corrupt user could flood the network with misleading messages (for example, a thousand messages that increment each of the analyzer's $d$ counters). But in practice, the analyzer can enforce\footnote{By placing the verification responsibility on the analyzer, we keep the shuffler lightweight. In particular, the shuffler does not have to keep track of the number of messages sent by users, which can be a sensitive attribute.} the communication constraint via a blind signature scheme \cite{Chaum82}: in a setup stage, each user interacts with the analyzer to sign exactly $k+1$ random strings. Each of the signed strings will serve as a tag of a message sent to the shuffler. The analyzer can limit its computation to messages with signed tags. This extra layer of security will only increase the communication cost by a small factor.

%%%%
\subsubsection{Comparison with $\cP_\had$}
A highlight of \cite{CSU19} is that two locally private protocols for mean estimation can have the same accuracy absent manipulation but one can be more robust to manipulation than the other. In the same spirit, we show that another shuffle protocol for histograms has roughly the same accuracy as $\cP_\flip$ absent manipulation but is less robust to manipulation.

We will consider $\cP_\had=(\cR_\had,\cA_\had)$ from Ghazi et al. \cite{GGKPV19}. We provide formal pseudocode in Appendix \ref{apdx:attack} but sketch the ideas here. Each of the $k+1$ messages sent by a user is either an encoding of some value $j\in[d]$ or a sample from a distribution that serves to hide the encodings of user values. The encodings are based upon a public Hadamard matrix to optimize communication complexity (total number of bits sent by a user). When it encounters an encoding of $j$, the analyzer increments a counter for $j$. The approximate histogram is constructed by applying a linear function to the counters.

Ghazi et al. give the following result concerning accuracy and privacy. When $\eps = \Theta(1)$, the bound on maximum error is asymptotically identical with that of $\cP_\flip$.
\begin{thm}[From \cite{GGKPV19}]
\label{thm:had}
Fix any $\eps ,\delta<1$ and $n\in\N$. There exist parameter choices $k = \Theta(\tfrac{1}{\eps^2} \log \tfrac{1}{\eps\delta})$ and $\tau = \Theta(\log n)$ such that $\cP_\had$ is $(\eps,\delta)$-shuffle private for inputs from $\cX_d$ and reports an approximate histogram with maximum error $O(\tfrac{\log d}{n} + \tfrac{1}{\eps n} \sqrt{\log \tfrac{1}{\eps\delta} \log d})$ with probability $9/10$.
\end{thm}

In both $\cP_\flip$ and $\cP_\had$, $m$ corrupt users can only shift counters by an additive factor of $m(k+1)$. But because $k$ is larger in $\cP_\had$ than in $\cP_\flip$, each corrupted user has a greater impact in the former protocol than the latter.
\begin{clm}
\label{clm:had-attack}
Choose $k,\tau$ as in Theorem \ref{thm:had}. If there is a coalition of $m<n$ corrupt users $M\subset [n]$, then for any target value $j\in[d]$ there is an input $\vec{x}$ such that $\cP_\had$ produces an estimate of $\hist_j(\vec{x})$ with bias $\tfrac{m}{n} \cdot (k+1) = \Omega( \tfrac{m}{n}\cdot \tfrac{1}{\eps^2} \log \tfrac{1}{\eps\delta})$.
\end{clm}
We defer the proof to Appendix \ref{apdx:attack} for space.

%% file: heavyhitters.tex
\subsection{Approximating Top-$t$ from $\cP_\flip$}
\label{sec:heavy-hitters}

Given an approximate histogram---as guaranteed by $\cP_\flip$---one can easily approximate the top-$t$ items: output the top-$t$ in the approximate histogram. If the maximum error is $\alpha/2$, then the rank of elements with frequency $< \hist_{[t]}(\vec{x})-\alpha$ in the approximate histogram cannot exceed $t$. Thus, the following is immediate from our earlier results.%Theorem \ref{thm:maximum}.
\begin{coro}
\label{coro:heavy-hitters-1}
For any $\vec{x}\in \cX^n_d$, if we compute $\vec{z} \gets \cP_\flip(\vec{x})$ (using the same parameters $k,q$ as in Theorem \ref{thm:maximum}), then $\mathrm{top}_t(\vec{z})$ $\alpha$-approximates the top-$t$ items in $\vec{x}$ with probability $\geq 9/10$, where $$\alpha = \max\left(\frac{1}{n}\cdot \frac{e^\eps+1}{e^\eps-1} \cdot \sqrt{\frac{1056}{5} \ln \frac{4}{\delta} \ln 20d} ,~ \frac{4}{n}\cdot \ln 20d \right) \cdot f(k).$$
\end{coro}

\begin{coro}
\label{coro:heavy-hitters-2}
For any $\vec{x}\in \cX^n_d$ and parameters $k,q$ such that $\frac{1}{nk+n} \ln 20d \leq q < 1/2$, if we compute $\vec{z} \gets \cP_\flip(\vec{x})$, then $\mathrm{top}_t(\vec{z})$ $\alpha$-approximates the top-$t$ items in $\vec{x}$ with probability $\geq 9/10$, where $$\alpha = 4 \sqrt{\frac{k+1}{n} q(1-q) \ln 20d } \cdot \left(\frac{1}{1-2q} \right)$$
\end{coro}

%% file: reduce_comm.tex
\section{Reducing Communication Complexity}
\label{sec:reduce-comm}

Although $\cP_\flip$ has a constant message complexity for a large range of $n$, each message is a binary string of length $d$. The communication complexity therefore grows with the dimension. In this section, we describe how to mitigate the impact of large dimension. %One is suited for the case where $n$ is small, the other when $n$ is large.

\subsection{Replacing Binary Strings with Lists}
In this subsection, we use the observation that the messages are binary strings that are likely sparse so that they can be equated with a short list of indices. More precisely, let $\cR_{\flip 2}$ be the local randomizer that, on input $x$, computes messages from $\cR_{\flip }(x)$ but replaces each binary string it creates with a list of the indices that contain bit 1. Let $\cA_{\flip 2}$ be the analyzer that converts each of the messages output by the shuffler back into a binary string and then runs $\cA_{\flip }$.

\begin{thm}
\label{thm:replacement}
If parameters $k,q$ are chosen in the same manner as in Theorem \ref{thm:maximum}, then $\cP_{\flip 3}=(\cR_{\flip 3}, \cA_{\flip 3})$ has the same number of messages and accuracy as $\cP_{\flip }$ but now the expected length of each message is $\leq \log_2 d \cdot (1 + dq) = O( \log d ( 1 + \tfrac{d}{\eps^2 n}  \log \tfrac{1}{\delta} + \tfrac{d}{nk}\log d ) )$ bits.
\end{thm}
\begin{proof}
A message is generated from either $\cR_{d,q}(0^{d})$ or $\cR_{d,q}(x \in \cX_{d})$. By construction, $\cR_{d,q}(0^{d})$ produces a string where each bit is drawn from $\Ber(q)$. This means the number of 1s is drawn from $\Bin(d,q)$. Meanwhile, the number of 1s generated from executing $\cR_{d,q}(x \in \cX_{d})$ is a sample from $\Ber(1-q) + \Bin(d-1,q)$. 

Recall we set $q$ to be $q = O(\tfrac{1}{\eps^2 n} \log \tfrac{1}{\delta} + \tfrac{1}{nk}\log d)$. This means the expected number of 1s in any message is $O(\tfrac{d}{\eps^2 n} \log \tfrac{1}{\delta} + \tfrac{d}{nk}\log d)$. And we need $\log_2 d$ bits to represent each index.
\end{proof}

\subsection{An Adaptation of Count-Min}
The change-of-representation in the preceding section is powerful when $n$ approaches (or exceeds) $d$. This subsection describes a method to reduce the communication complexity when $n$ is not so large, at the price of logarithmic message complexity. The new protocol, which we call $\cP_{\flip 3}$, uses the randomizer and analyzer of $\cP_{\flip 2}$ as black boxes. Based upon the Count-Min technique from the sketching literature, $\cP_{\flip 3}$ is an instance of a general method of transforming any shuffle protocol for histograms.

The pseudocode for the randomizer and analyzer is given in Algorithms \ref{alg:r-flip3} and \ref{alg:a-flip3}, respectively. The heart of the transformation is hashing the universe $[d]$ to $[\hat{d}]$. If an element $j$ experiences no collisions, note that its frequency in the hashed dataset is the same as in the original dataset. Otherwise, the frequency is an overestimate. When there are many hash functions, it is likely that there is some hash function where $j$ experiences no collisions; taking the minimum over the frequencies in the hashed datasets would recover the original frequency. We execute $\cP_{\flip 2}$ once for each hashed datset to obtain estimates of the frequencies.

\begin{algorithm}
\caption{$\cR_{\flip 3}$ a local randomizer for histograms}
\label{alg:r-flip3}

\KwIn{$x\in \cX_d$; parameters $V,\hat{d},k\in\N, q\in (0,1/2)$}
\KwOut{$\vec{y}\in ([V]\times [\hat{d}]^* )^*$}

Obtain hash functions $\{h^{(v)} : \cX_d \to \cX_{\hat{d}}\}$ from public randomness.

Initialize $\vec{y}$ to the empty vector

\For{$v\in[V]$}{

	Compute messages $\vec{y}^{(v)} \gets \cR_{\flip 2} (h^{(v)}(x))$ using dimension $\hat{d}$

	\For{$y \in \vec{y}^{(v)}$} {
		Append labeled message $(v,y)$ to $\vec{y}$	
	}
}

\Return{$\vec{y}$}
\end{algorithm}

\begin{algorithm}
\caption{$\cA_{\flip 3}$ an analyzer for histograms}
\label{alg:a-flip3}

\KwIn{$\vec{y}\in ([V]\times [\hat{d}]^* )^*$; parameters $V,\hat{d},k\in\N, q\in(0,1/2)$}
\KwOut{$\vec{z}\in\R^d$}

Obtain hash functions $\{h^{(v)} : \cX_d \to \cX_{\hat{d}}\}$ from public randomness.

\For{$j\in[d]$}{
	$z_j\gets \infty$
}

\For{$v\in[V]$}{
	Initialize $\vec{y}^{(v)}\gets \emptyset$
	
	\For{$(v',y) \in \vec{y}$} {
		Append message $y$ to $\vec{y}^{(v)}$ if label $v'$ matches $v$
	}
	
	Compute $\hat{z}^{(v)} \gets \cA_{\flip 2}(\vec{y}^{(v)})$ using dimension $\hat{d}$
	
	\For{$j\in[d]$}{
		$\hat{j} \gets h^{(v)}(j)$
	
		$z_j\gets \min(z_j,\hat{z}^{(v)}_{\hat{j}})$
	}
}

\Return{$\vec{z}$}
\end{algorithm}

We first analyze the protocol in terms of the parameter $V$, which determines the number of hash functions and protocol repetitions. We will choose a value for $V$ later in the section.

\begin{clm}
\label{clm:count-min}
Fix any $\eps = O(1)$, $\delta < 1/100$, and number of users $n \in \N$. If $\hat{d}\gets \lceil n\cdot (100d)^{1/V} \rceil$ and $k > \max \left( \tfrac{134}{5n} (\tfrac{e^\eps+1}{e^\eps-1})^2 \ln \tfrac{4}{\delta}, \tfrac{2}{n}\ln20\hat{d}V -1 \right)$, then there is a choice of parameter $q<1/2$ such that $\cP_{\flip2}$ has the following properties:
\begin{enumerate}[a.]
\item
    Each user sends $V\cdot (k + 1 )$ messages, each consisting of $O(\log V + \log \hat{d} ( 1 + \tfrac{\hat{d}}{\eps^2 n}  \log \tfrac{1}{\delta} + \tfrac{\hat{d}}{nk}\log \hat{d} ) )$ bits in expectation.
\item
	$\cP_{\flip2}$ is $\left( \eps (e^{\eps}-1)\cdot V + \eps \cdot \sqrt{2V \log \tfrac{1}{V\delta}}, 2V\delta \right)$-shuffle private for inputs from $\cX_d$
\item
	For any $\vec{x}\in \cX^n_d$, $\cP_{\flip2}(\vec{x})$ reports approximate histogram $\vec{z}$ such that the maximum error is
	\begin{align*}
	\norm{\vec{z} - \hist(\vec{x})}_\infty &= O\left( \frac{1}{\eps n} \sqrt{ \log \frac{1}{\delta} \log \hat{d}V } + \frac{ \log \hat{d}V}{n} \right) %\\
	%&= O\left( \frac{1}{\eps n} \sqrt{ \log \frac{1}{\delta} (\log nV + (\log d)/V) } + \frac{ \log nV + (\log d)/V}{n} \right)
	\end{align*}
	with probability $\geq 9/10-(1/100)^V$.
\end{enumerate}
\end{clm}
\begin{proof}
We will choose $q$ in much the same way as Theorems \ref{thm:maximum} and \ref{thm:replacement}. The sole modification is that we change the $\ln 20d$ term to $\ln 20 \hat{d} V$.

The protocol executes $\cP_{\flip 2}$ exactly $V$ times using the hashed dimension $\hat{d}$. So the number of messages is simply $V \cdot (k+1)$. Each message is generated via $\cR_{\flip 2}$ and then labeled by the execution number $v$, so Part \textit{a} is immediate from Theorem \ref{thm:replacement}. Meanwhile, Part \textit{b} follows directly from advanced composition (Fact \ref{fact:adv_comp}).

To prove Part \textit{c}, we first analyze the randomness from hashing and consider privacy noise later. For any $j\in[d]$, let $E_j$ denote the event that there is at least one hash function where $j$ experiences no collisions with a user value $j'\neq j$. Formally, $\exists v^* ~ \forall j'\in \vec{x},j'\neq j~ h^{(v^*)}(j)\neq h^{(v^*)}(j')$. We will now bound the probability that $E_j$ does not occur.
\begin{align*}
\pr{\vec{h}}{\neg E_j} &= \pr{\vec{h}}{\forall v ~ \exists j' \in \vec{x}~ h^{(v)}(j)= h^{(v)}(j')} \\
	&= \pr{\vec{h}}{\exists j' \in \vec{x}~ h^{(v)}(j)= h^{(t)}(j')}^V \\
	&\leq \left( n\cdot \pr{\vec{h}}{h^{(v)}(j)= h^{(v)}(j')} \right)^V \\
	&= (n/\hat{d})^V = (1/100)^V \cdot \frac{1}{d}
\end{align*}
By a union bound, the probability that there is some $j$ where $E_j$ does not occur is at most $(1/100)^V$

The remainder of the proof conditions on $E_j$ occurring for all $j$. In this event, each $j$ can be paired with some $v^*$ where the count of $h^{(v^*)}(j)$ in the hashed dataset is exactly the count of $j$ in the original dataset. For any $v\neq v^*$, observe that the count of $h^{(v)}(j)$ in the hashed dataset is must be either (1) an overestimate due to collision or (2) also equal. Thus, the minimum over the counts yields the correct value.

Now we incorporate the fact that $\cA_{\flip 3}$ only has private estimates of the counts. When we set $\beta=1/10\hat{d}V$, Claim \ref{clm:acc-per-bin} and a union bound imply each protocol execution has $\ell_\infty$ error $$2 \sqrt{\frac{k+1}{n} q(1-q) \ln 20\hat{d}V } \cdot \left(\frac{1}{1-2q}\right)$$ except with probability $\leq 1/10V$ A second union bound over the $V$ executions ensures that the privatized count of $h^{(v^*)}(j)$ in the hashed dataset is the minimum of all privatized counts. Substitution of $q$ completes the proof.
\end{proof}

We now show that there is a choice of $V$ where the expected communication complexity has only a polylogarithmic dependence on $d$ and $n$.
\begin{thm}
\label{thm:count-min}
Fix any $\eps = O(1)$. If $n = \Omega( \tfrac{\log d}{{\eps}^2} \log \tfrac{1}{{\delta}} \log \tfrac{\log d}{{\delta}})$ and $\delta = O(1/n)$, there are choices of parameters $V,k \in \N$ and $q<1/2$ such that $\cP_{\flip 3}$ has the following properties:
\begin{enumerate}[a.]
\item
    Each user sends $2\log_2 d$ messages, each consisting of $O\left( \tfrac{1}{\eps^2 }  \log d \log^3\tfrac{\log d}{\delta} \right)$ bits in expectation.
\item
	$\cP_{\flip 3}$ is $\left({\eps}, {\delta} \right)$-shuffle private for inputs from $\cX_d$
\item
	For any $\vec{x}\in \cX^n_d$, $\cP_{\flip 3}(\vec{x})$ reports approximate histogram $\vec{z}$ such that the maximum error is
	\begin{align*}
	\norm{\vec{z} - \hist(\vec{x})}_\infty 
	&= O\left(\frac{1}{\eps n} \sqrt{\log d } \log^{3/2} \left( \frac{\log d}{\delta} \right) \right)
	\end{align*}
	with probability $\geq 9/10-(1/100)^{\log_2 d}$.
\end{enumerate}
\end{thm}
\begin{proof}
We will set $V \gets \log_2 d$. Let $\overline{\eps} = {\eps}/ 4\sqrt{V \log (1/{\delta})}$ and $\overline{\delta} = \delta / 2V$. Because $n$ is sufficiently large, it is possible to set $k=1$ and set $q \in (0,1/2)$ such that each execution of $\cP_{\flip 2}$ satisfies $(\overline{\eps}, \overline{\delta} )$-shuffle privacy. By substitution into Claim \ref{clm:count-min}, the expected length of a message is 
\begin{align*}
&O\left( \log V + \log \hat{d} \left( 1 + \frac{\hat{d}}{\overline{\eps}^2 n}  \log \frac{1}{\overline{\delta}} + \frac{\hat{d}}{nk}\log \hat{d} \right) \right) \\
={}& O\left( \log V + \log n \left( \frac{1}{\overline{\eps}^2 }  \log \frac{1}{\overline{\delta}} + \log n \right) \right) \tag{Choice of $\hat{d},k$} \\
={}& O\left( \log V + \log n \left( \frac{1}{\eps^2 }  V\log \frac{1}{\delta} \log \frac{V}{\delta} + \log n \right) \right) \tag{Choice of $\overline{\eps},\overline{\delta}$}\\
={}& O\left( \log \log d + \log n \left( \frac{1}{\eps^2 }  \log d \log \frac{1}{\delta} \log \frac{\log d}{\delta} + \log n \right) \right) \tag{Choice of $V$} \\
={}& O\left( \frac{1}{\eps^2 }  \log d \log^3\frac{\log d}{\delta} \right) \tag{$\delta = O(1/n)$}
\end{align*}

Meanwhile, the maximum error is
\begin{align*}
&O\left( \frac{1}{\overline{\eps} n} \sqrt{ \log \frac{1}{\overline{\delta}} \log \hat{d}V } + \frac{ \log \hat{d}V}{n} \right)\\
={}& O\left( \frac{1}{\overline{\eps} n} \sqrt{ \log \frac{1}{\overline{\delta}} (\log n + \log \log d) } + \frac{ \log n + \log \log d}{n} \right) \tag{Choice of $\hat{d},V$}\\
={}& O\left( \frac{1}{\eps n} \sqrt{\log d \log \frac{1}{\delta} \log \frac{\log d}{\delta} (\log n + \log \log d) } + \frac{ \log n + \log \log d}{n} \right) \tag{Choice of $\overline{\eps},\overline{\delta}$}\\
={}& O\left( \frac{1}{\eps n} \sqrt{\log d } \log^{3/2} \left( \frac{\log d}{\delta} \right) \right) \tag{$\delta = O(1/n)$}
\end{align*}
This completes the proof.
\end{proof}

Finally, we study the impact that corrupt users can have on the estimates generated by $\cP_{\flip 3}$.
\begin{clm}
Fix any $\eps=O(1)$, $\delta<1/100$, and $n,q,V$ as in Theorem \ref{thm:count-min}. For any input $\vec{x}\in\cX^n_d$, any coalition of $m$ corrupt users can introduce $ O( \tfrac{m}{n} \cdot \log d )$ error to $\cP_{\flip 3}$.
\end{clm}
\begin{proof}[Proof Sketch]
Each honest user transmits $O(V)$ messages, where each message is an output of $\cR_\flip \circ h^{(v)}$ labeled by $v$. A corrupt user is therefore limited a ``budget'' of $O(V)$ messages each with the same structure. But all of these messages could share the same label $v$. In this case, we can adapt the analysis of $\cP_\flip$ to show that $m$ corrupt users will add $O(\tfrac{m}{n}\cdot V)$ bias to the $v$-th execution of $\cP_{\flip}$. %The other executions will have up to $\tfrac{m}{n}\cdot (k+1) \cdot f(k)$ bias.
\end{proof}

%Since $q \propto 1/n$, the improvement is noticeable for large $n$.

%\paragraph{Alternative compression methods.} The length of each message in $\cP_{\flip 3}$ is $O(n + \log \log d)$, which is not an improvement when $n\geq d$. In this large $n$ regime, 

%Another method of reducing the communication complexity comes from the fact that each message produced by $\cR_\flip$ is made with randomized response. Recent work by Feldman and Talwar \cite{FT21} shows how to use a pseudorandom generator to simulate a locally private algorithm, then argues sending the short seed of the generator is also locally private. Given that randomized response is a locally private algorithm, the method applies to $\cR_\flip$.

%% file: experiments.tex
\section{Experiments}
\label{sec:experiments}

In this section, we evaluate the accuracy of our protocol on natural language data. To give context for these results, we repeat the experiment on the histogram protocol by Balcer \& Cheu \cite{BC20}. It has essentially the same communication complexity as $\cP_\flip$, but its message complexity is $O(d)$.

\medskip

We acquire a list of $d\approx 4.7\cdot 10^5$ English words from a publicly accessible repository and $\approx 3.7\cdot 10^6$ tweets on Twitter in the United States previously used in work by Cheng, Caverlee, and Lee \cite{CCL10}.\footnote{The word list was downloaded from https://github.com/dwyl/english-words while the tweets were downloaded from https://archive.org/details/twitter\_cikm\_2010} We sampled one recognized word from each tweet, so that $n\approx 3.7\cdot 10^6$. Fixing privacy parameters $\eps=1$ and $\delta=10^{-7}$, we simulated our protocol on the dataset a hundred times for four choices of $k$.

As an aside, our method of sampling data ensures tweet-level privacy rather than user-level privacy. That is, we could have sampled one word per user instead of one word per tweet. But our goal was to evaluate the protocol when applied to large-scale data analysis and our dataset consists of only $\approx 10^4$ users. %will not result in substantially altered outcomes, as the distribution would still be Zipfian and we would be running the protocol with the same parameters.

\subsection{Evaluation of Maximum Error}

In Figure \ref{fig:max_experiments}, we visualize the error introduced by $\cP_\flip$ for varying choices of $k$. We also plot corresponding confidence bounds derived from Corollary \ref{coro:acc-max}. As predicted, the error decreases with larger $k$. And, at least for this particular dataset, our bounds are loose by only a small multiplicative factor.

\begin{figure}[h]
    \centering
    \includegraphics[width=0.75\textwidth]{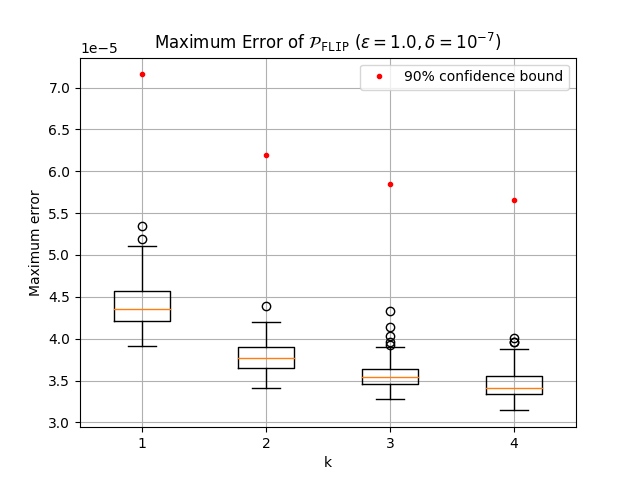}
    \caption{Maximum error of frequency estimates in experiments, as a (decreasing) function of $k$. Confidence bounds (red filled circles) are derived from Corollary \ref{coro:acc-max}.}
    \label{fig:max_experiments}
\end{figure}

In Figure \ref{fig:bc_max_experiments}, we compare the max error of $\cP_\flip$ with the histogram protocol by Balcer \& Cheu \cite{BC20}. The primary advantage of \cite{BC20} over $\cP_\flip$ is that the error introduced to any bin does not scale with $d$. Specifically, the maximum error is $O(\tfrac{1}{\eps^2 n}\log \tfrac{1}{\delta})$. In contrast $\cP_\flip$ only ensures $O(\tfrac{\log d}{n} + \tfrac{1}{\eps n} \sqrt{\log d \log \tfrac{1}{\delta}})$ error. But in our application, $d$ is actually orders of magnitude smaller than $1/\delta$ so $\cP_\flip$ is in fact more accurate. 

\begin{figure}[h]
    \centering
    \includegraphics[width=0.75\textwidth]{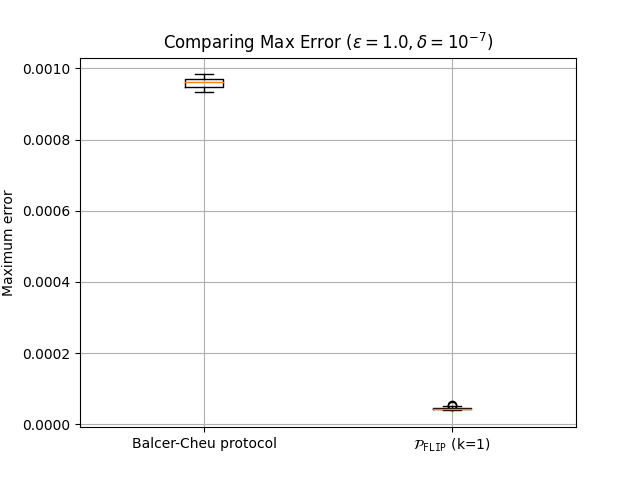}
    \caption{Comparison between the maximum error of $\cP_\flip$ and that of the protocol in \cite{BC20}. The latter protocol would only have better error than the former if the dimension of the data were much larger.}
    \label{fig:bc_max_experiments}
\end{figure}

%: over a hundred experiments, the maximum error of \cite{BC20} was never less than $2\cdot10^{-4}$ as compared to the median of $\approx 4.4 \cdot 10^{-5}$ achieved by $\cP_\flip$ ($k=1$).

%%%%
\subsection{Evaluation of Top-$t$ selection}
Recall the simple top-$t$ selection strategy from Section \ref{sec:heavy-hitters}: report the top-$t$ items of a private version of the histogram. In Table \ref{tab:gaps}, we fix $t=6000$ and present our bound from Corollary \ref{coro:heavy-hitters-2} alongside the maximum observed value in our experiments. We remark that the frequency of the rank-$t$ word is $1.33\cdot10^{-5}$. This is an upper bound on $\alpha$ that holds with probability 1, since the worst that can happen is that a word with frequency 0 displaces the $t$-th most common word.

\begin{table}[h]
    \centering
    \begin{tabular}{c|c|c|}
        & \multicolumn{2}{c|}{$\cP_\flip$'s $\alpha$-approximation of top-6000} \\
        $k$ & Bound from Corollary \ref{coro:heavy-hitters-2} & Maximum observed \\ \hline
        1 & $1.43\cdot10^{-4}$ & \\
        2 & $1.24\cdot10^{-4}$ & \multirow{2}{*}{$1.33\cdot10^{-5}$}\\
        3 & $1.17\cdot10^{-4}$ & \\
        4 & $1.13\cdot10^{-4}$ & 
    \end{tabular}
    \caption{Comparing the bound on the error of top-$t$ selection with experimental results.}
    \label{tab:gaps}
\end{table}

In Figure \ref{fig:f1_experiments}, we plot the F1 score of the report-top-$t$ strategy for both $\cP_\flip$ and the Balcer-Cheu protocol. $\cP_\flip$ preserves $\approx 95\%$ of the top-2000 words in the dataset and is consistently more accurate than the alternative protocol. Increasing $t$ decreases the F1 score for both protocols because infrequent words are easily evicted from the top-$t$ and natural language heavily favors a small set of words.

\begin{figure}[h]
    \centering
    \includegraphics[width=0.75\textwidth]{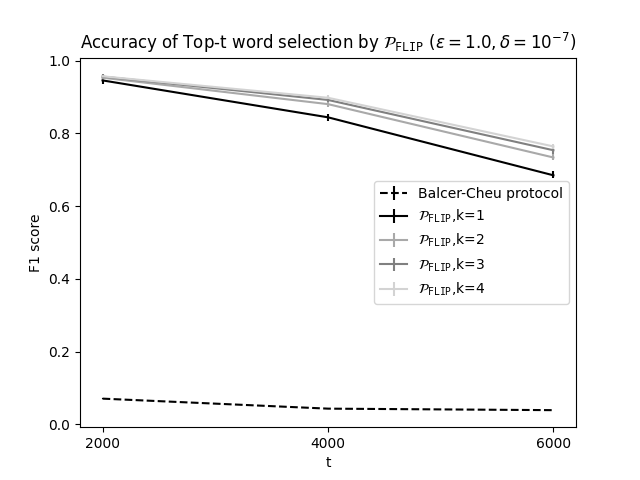}
    \caption{F1 scores for top-$t$ word selection, for both tested protocols. The lines connect medians while the error bars represent the complete range of observed values.}
    \label{fig:f1_experiments}
\end{figure}

%$\cP_\flip$ also has the advantages of much smaller message complexity ($O(1)$ vs. $\Omega(d)$) and slightly smaller communication complexity ($O(d)$ vs. $\Omega(d \log d)$).

%% file: tpdp_techniques.tex
\section{Our Techniques}
Each user $i$ in our main protocol first encodes their data $x_i\in[d]$ as a binary string with a single 1 bit at index $x_i$ (a ``one-hot'' string). The user then flips each bit independently with some fixed probability $q$. Next, they create $k$ other strings that are all 0 and then repeat this bit flipping, which amounts to injecting $k$ fake users to the protocol.

We choose $q$ so that the messages from the $nk$ fake users provide differential privacy for the real users. The first step in the analysis is to observe that for any $x\neq x' \in [d]$, the corresponding one-hots only differ on two indices (all other positions are guaranteed to be 0). Because differential privacy is closed under post-processing, it therefore suffices to consider the case where $d=2$. Next, we observe the probability of generating the string 00 (resp. 11) when bit flipping the encoding of $x \in [2]$ is the same as the probability of generating the string 00 (resp. 11) when bit flipping the encoding of $x' \in [2]$. Again using closure under post-processing, it suffices to study the privacy guarantee of the algorithm $\cB$ that takes one message $\in \{01,10\}$ and shuffles it with messages generated by $nk$ fake users.

The probability that a fake user outputs a message $\in \{01,10\}$ is exactly $2q(1-q)$. By concentration of the binomial distribution, the number of such messages is $\Theta(nkq(1-q))$ with $1-\delta$ probability. By symmetry, whether such a message is 01 or 10 is a $\Ber(1/2)$ event. So the count of $01$ is drawn from $\Bin(\Theta(nkq(1-q)),1/2)$ with $1-\delta$ probability. And the literature has established that the binomial mechanism is $(\eps,\delta)$-differentially private so long as the variance is $\Omega(\tfrac{1}{\eps^2}\log \tfrac{1}{\delta})$.\footnote{For recent analysis of binomial noise in the shuffle model, see Ghazi et al. \cite{GGKPV19} and Cheu et al. \cite{CSUZZ19}.} This allows us to derive bounds on $k$ and $q$. We note that our full work contains an alternative analysis of $\cB$ which that yields improved constants.

The analyzer obtains an estimate of bin $j$ by de-biasing and re-scaling the sum of the $j$-th bits in the message strings. The bias is due to the noise introduced by fake users. Re-scaling is necessary because the one-hot encodings are randomly flipped, dampening the signal. Because $q$ can be viewed as a term $q_k$ that depends on $k$, the scaling factor is $f(k) = 1/(1-2q_k)$. Since each user is allotted $k+1$ messages and each contributes $f(k)$ to the sum computed by the analyzer, $m$ corrupt users cannot skew an estimate beyond $m\cdot (k+1)\cdot f(k)$.

\paragraph{Reducing Communication Complexity}
    Our technique to reduce communication complexity proceeds in two stages. We first make the simple observation that a binary string with known length is equivalent to a list of the indices where the string has value 1. By construction, a message generated by our local randomizer is a binary string where the number of such indices has expectation $O(d q)$. Our choice of $q$ is proportional to $1/n$, so this alternative representation is very effective when $n$ approaches or exceeds $d$.
    
    The small $n$ regime motivates a second round of compression. We describe an adaptation of the count-min sketching technique. Given a uniformly random hash function, we can reduce the size of the domain $d$ to some $\hat{d}$ at the cost of some collisions. We repeatedly hash in order to reduce the likelihood of error due to collisions and run our histogram protocol on the hashed data. We remark that Ghazi et al. \cite{GGKPV19} build a specific histogram protocol out of count-min, while we use it as a tool that can improve the communication complexity of \emph{arbitrary} histogram protocols.

%% file: appendix/technical.tex
\section{Technical Claims for $\cP_\flip$}
\label{apdx:technical}

\begin{clm*}[Restatement of Claim \ref{clm:noise-concentration}]
Fix $m\in \N$ and $q, \delta \in (0,1)$. Define
\begin{align*}
\Delta &:= \sqrt{3mq(1-q)\ln \frac{4}{\delta}} \cdot \frac{q(1-q)}{1-q(1-q)} \\
U &:= mq(1-q) + \Delta + \sqrt{3(mq(1-q) +\Delta) \ln\frac{4}{\delta}} \\
L &:= mq(1-q) - \Delta - \sqrt{3(mq(1-q) +\Delta) \ln\frac{4}{\delta}}
\end{align*}
Let $F\subset \Z^4$ denote the set of vectors where $\vec{t} \in F$ if and only if $t_2,t_3 \in [L,U]$. If $mq(1-q) > \frac{9}{2}\ln (4/\delta)$, then
    \begin{equation*}
    \pr{\vec{f}\sim \cM(m,q)}{ \vec{f} \notin F} \leq \delta
    \end{equation*}
\end{clm*}

\begin{proof}
We will use a Chernoff bound to argue that the marginal distribution of $f_3$ is likely to be in some interval $[L',U']$. Then we will use a Chernoff bound to argue that the distribution of $f_2$ conditioned on $f_3\in [L',U']$ is likely to be in $[L,U]$. The claim follows from the fact that $[L',U'] \subset [L,U]$.

By construction, $f_3$ is the random variable that counts the number of times the message 10 (2 in binary) is produced by $m$ executions of $\cR_{2,q}(00)$. Referring to Table \ref{tab:pmf}, this means $f_3$ is distributed as $\Bin(m, q(1-q))$. Using $\mu_3$ as shorthand for the mean $mq(1-q)$, multiplicative Chernoff bounds imply the following for all $z \in (0,1)$: $$\pr{}{|f_3 - \mu_3| > z\mu_3} \leq 2\exp(- z^2\mu_3/3).$$ Because $\mu_3 \ge 3 \ln\frac{4}{\delta}$, we can assign $z \gets \sqrt{\frac{3}{\mu_3} ln\frac{4}{\delta}}$ so that $$\pr{}{|f_3 - \mu_3| > \sqrt{3 \mu_3 \cdot \ln\frac{4}{\delta}} } \leq \delta/2$$
So if we define $L' \gets \mu_3 - \sqrt{3 \mu_3 \cdot \ln\frac{4}{\delta}}$ and $U' \gets \mu_3 + \sqrt{3 \mu_3 \cdot \ln\frac{4}{\delta}}$, $f_3 \in [L',U']$ except with probability $\delta/2$; the remainder of the proof conditions on this event.

Specifically, we assume that the random variable $f_3$ takes on some value $r \in [L,U]$. This means that $f_2$ is the random variable that counts the number of times the message 01 (1 in binary) is produced by $m-r$ executions of $\cR_{2,q}(00)$ \emph{conditioned on the output not being 10 (2 in binary)}. Referring to Table \ref{tab:pmf}, this means $f_2$ is distributed as $\Bin\left( m-r, \frac{q(1-q)}{1-q(1-q)} \right)$. The mean of this distribution is $\mu_2 \gets (m-r) \cdot \frac{q(1-q)}{1-q(1-q)}$. If we could show $\mu_2 \geq 3\ln \frac{4}{\delta}$, we could again invoke multiplicative Chernoff bounds to argue $$\pr{}{|f_2 - \mu_2| > \sqrt{3 \mu_2 \cdot \ln\frac{4}{\delta}} } \leq \delta/2.$$

Notice that $r \in [L', U']$ implies
\begin{align*}
\mu_2 =&{} (m-r) \cdot \frac{q(1-q)}{1-q(1-q)} \\
    \geq{}& \left( m-\mu_3 - \sqrt{3\mu_3 \cdot \ln \frac{4}{\delta}} \right) \cdot \frac{q(1-q)}{1-q(1-q)} \\
    ={}& \left( m-mq(1-q)-\sqrt{3mq(1-q)\ln \frac{4}{\delta}} \right) \cdot \frac{q(1-q)}{1-q(1-q)} \\
    ={}& mq(1-q) - \sqrt{3mq(1-q)\ln \frac{4}{\delta}} \cdot \frac{q(1-q)}{1-q(1-q)} \\
    ={}& mq(1-q) - \Delta
\end{align*}
By symmetric arguments, $$ \mu_2 \leq mq(1-q) +\Delta.$$ The claim follows by substitution.

We now argue that $\mu_2 \geq 3\ln \frac{4}{\delta}$.
\begin{align*}
\mu_2 \geq{}& mq(1-q) - \sqrt{\frac{1}{3}\cdot mq(1-q)\ln \frac{4}{\delta}} \tag{$q(1-q)<1/4$}\\
    \geq{}& \frac{2}{3} mq(1-q) \tag{$mq(1-q) > 3\ln(4/\delta)$} \\
    \geq{}& 3\ln \frac{4}{\delta}
\end{align*}

\end{proof}

\begin{clm*}[Restatement of Claim \ref{clm:good-noise}]
Fix any $\eps>0$ and $\delta < 1/100$. Define $F$ as in Claim \ref{clm:noise-concentration}. If $q< 1/2$ and $mq(1-q) \geq \tfrac{33}{5} \left(\tfrac{e^\eps+1}{e^\eps-1}\right)^2 \ln(4/\delta)$, then for any $\vec{y} = (y_1,\dots,y_4)$,
\begin{align}
    \pr{\vec{f} \gets \cM(m,q) }{\vec{f} = (y_1,y_2-1,y_3,y_4) ,\vec{f} \in F } &\leq e^\eps \cdot \pr{\vec{f} \gets \cM(m,q) }{\vec{f} = (y_1,y_2,y_3 - 1,y_4) } \label{eq:good-noise}\\
    \pr{\vec{f} \gets \cM(m,q) }{\vec{f} = (y_1,y_2,y_3 - 1, y_4) ,\vec{f} \in F } &\leq e^\eps \cdot \pr{\vec{f} \gets \cM(m,q) }{\vec{f} = (y_1,y_2 - 1,y_3,y_4) } \label{eq:good-noise-sym}
\end{align}
\end{clm*}

\begin{proof}
It remains to prove \eqref{eq:good-noise}; the proof of \eqref{eq:good-noise-sym} will be completely symmetric. Let $F' \subset \Z^4 $ denote the set of vectors where $\vec{f}\in F'$ if and only if $f_2,f_3 \in [L-1,U+1]$.
\begin{align*}
    &\pr{\vec{f} \gets \cM(m,q) }{\vec{f} = (y_1,y_2-1,y_3,y_4) ,\vec{f} \in F } \\
    ={}& \pr{\vec{f} \gets \cM(m,q) }{\vec{f} = (y_1,y_2-1,y_3,y_4)} \cdot \indic{(y_1,y_2-1,y_3,y_4) \in F} \\
    ={}& \frac{m!}{y_1!(y_2-1)!y_3!y_4!} \cdot (1-q)^{2y_1}(q(1-q))^{y_2-1}(q(1-q))^{y_3}q^{2y_4} \\
        & \cdot \indic{(y_1,y_2-1,y_3,y_4) \in F} \tag{Defn. of $\cM$} \\
    \leq{}& \frac{m!}{y_1!(y_2-1)!y_3!y_4!} \cdot (1-q)^{2y_1}(q(1-q))^{y_2-1}(q(1-q))^{y_3}q^{2y_4} \\
        & \cdot \indic{(y_1,y_2,y_3-1,y_4) \in F'} \stepcounter{equation} \tag{\theequation} \label{eq:good-noise-1}
\end{align*}
\eqref{eq:good-noise-1} comes from the fact that when $y_2-1 \in [L,U]$ and $y_3 \in [L,U]$, it must be the case that $y_2 \in [L-1,U+1]$ and $y_3-1 \in [L-1,U+1]$. We can also derive
\begin{align*}
    &\pr{\vec{f} \gets \cM(m,q) }{\vec{f} = (y_1,y_2,y_3 - 1,y_4), \vec{f} \in F' }\\
    ={}& \frac{m!}{y_1!y_2!(y_3-1)!y_4!} \cdot (1-q)^{2y_1}(q(1-q))^{y_2}(q(1-q))^{y_3-1}q^{2y_4} \\
    &\cdot \indic{(y_1,y_2,y_3-1,y_4) \in F'} \stepcounter{equation} \tag{\theequation} \label{eq:good-noise-2}
\end{align*}

By combining \eqref{eq:good-noise-1} and \eqref{eq:good-noise-2},
\begin{align*}
&\pr{\vec{f} \gets \cM(m,q) }{\vec{f} = (y_1,y_2-1,y_3,y_4) ,\vec{f} \in F } \\
\leq{}& \frac{y_2}{y_3} \cdot \pr{\vec{f} \gets \cM(m,q) }{\vec{f} = (y_1,y_2,y_3 - 1,y_4) } \cdot  \indic{(y_1,y_2,y_3-1,y_4) \in F'}.
\end{align*}

In the case where $(y_1,y_2,y_3-1,y_4) \notin F'$, the right hand side is zero so that \eqref{eq:good-noise} trivially holds. Otherwise, $y_2 / y_3 \leq (U+1) / (L-1)$ by definition of $F'$. This means $$\pr{\vec{f} \gets \cM(m,q) }{\vec{f} = (y_1,y_2-1,y_3,y_4) ,\vec{f} \in F } \leq \frac{U+1}{L-1} \cdot \pr{\vec{f} \gets \cM(m,q) }{\vec{f} = (y_1,y_2,y_3 - 1,y_4) }$$ so it simply remains to show $(U+1)/(L-1) \leq e^\eps$. We rewrite this target inequality as
\begin{align*}
\frac{e^\eps - 1}{e^\eps + 1} &\geq \frac{\Delta + \sqrt{3(mq(1-q)+\Delta) \ln (4/\delta)} + 1}{mq(1-q)} \\
    &= \underbrace{\sqrt{\frac{3\ln(4/\delta)}{mq(1-q)}} \cdot \frac{q(1-q)}{1-q(1-q)}}_{A} + \underbrace{\frac{\sqrt{3(mq(1-q)+\Delta) \ln (4/\delta)}}{mq(1-q)}}_{B} + \underbrace{\frac{1}{mq(1-q)}}_{C} \stepcounter{equation} \tag{\theequation} \label{eq:good-noise-3}
\end{align*}
We will upper bound each term, beginning with $A$:
\begin{align*}
A &= \sqrt{\frac{3\ln(4/\delta)}{mq(1-q)}} \cdot \frac{q(1-q)}{1-q(1-q)} \\
    &< \sqrt{\frac{3\ln(4/\delta)}{mq(1-q)}} \cdot \frac{1}{3} \tag{$q < \half $} \\
    &< \frac{1}{3} \cdot\sqrt{\frac{5}{11}}\cdot \frac{e^\eps - 1}{e^\eps + 1}
\end{align*}

Now we bound $B$:
\begin{align*}
B &= \frac{\sqrt{3(mq(1-q)+\Delta) \ln (4/\delta)}}{mq(1-q)} \\
    &= \sqrt{\frac{3\ln(4/\delta)}{mq(1-q)} + \frac{3\Delta\ln(4/\delta)}{(mq(1-q))^2}} \\
    &= \sqrt{\frac{3\ln(4/\delta)}{mq(1-q)} + \frac{3\ln(4/\delta)}{(mq(1-q))^2} \cdot \left(\sqrt{3mq(1-q)\ln \frac{4}{\delta}} \cdot \frac{q(1-q)}{1-q(1-q)} \right)} \tag{Value of $\Delta$}\\
    &= \sqrt{\frac{3\ln(4/\delta)}{mq(1-q)} + \left(\frac{3\ln(4/\delta)}{mq(1-q)}\right) ^{3/2} \cdot \frac{q(1-q)}{1-q(1-q)} } \\
    &< \sqrt{\frac{5}{11} + \left(\frac{5}{11}\right)^{3/2} \cdot \frac{1}{3} } \cdot \frac{e^\eps - 1}{e^\eps + 1}
\end{align*}

Finally we bound $C$:
\begin{align*}
C = \frac{1}{mq(1-q)} &\leq \frac{5}{33\ln(4/\delta)} \cdot \frac{e^\eps - 1}{e^\eps + 1} \\
    &\leq \frac{5}{33\ln(400)} \cdot \frac{e^\eps - 1}{e^\eps + 1} \tag{$\delta \leq 1/100$}
\end{align*}
\eqref{eq:good-noise-3} follows by substitution.
\end{proof}

%% file: appendix/attack.tex
\section{Manipulation Attack Against $\cP_\had$}
\label{apdx:attack}
In this section, we describe the Hadamard response protocol by Ghazi et al. \cite{GGKPV19} and a manipulation attack against it. For a wide range of $n$, the protocol's estimates are less robust (at least in the worst case) than $\cP_\flip$.

We present pseudocode for the randomizer and analyzer in Algorithms \ref{alg:r-had} and \ref{alg:a-had}, which use parameters $k,\tau\in\N$. We remark that we have adjusted the algorithm and notation to be more consistent with our protocol and the problem it solves. Specifically, parameter $\rho$ is renamed $k$ to match  $\cP_\flip$ and we limit user data to $\cX_d$.\footnote{As originally written, $\cP_\had$ solved the more general problem where users can have more than one item $\in[d]$. In principle, we could augment $\cP_\flip$ to solve the same generalization, but we focus on the simplest case for clarity.}

\begin{algorithm}
\caption{$\cR_\had$, local randomizer for histograms}
\label{alg:r-had}

\KwIn{$x \in \cX_d$}
\KwOut{$\vec{y} \in ([2d]^\tau)^{k+1}$}

Initialize $\vec{y}$ to the empty vector.

Let $j(x)$ be the integer $j$ such that  $e_{j,d}=x$

Let $h_{j(x)}$ be the $j(x)+1$-th row of the $2d\times 2d$ Hadamard matrix

Sample $a_1,\dots, a_\tau$ uniformly and independently from $\{\hat{j} ~|~ h_{j(x),\hat{j}} = 1\}$

Append the tuple $(a_1,\dots, a_\tau)$ to $\vec{y}$

\For{$i\in[k]$}{
    Sample $a_1,\dots, a_\tau$ uniformly and independently from $[2d]$
    
    Append the tuple $(a_1,\dots, a_\tau)$ to $\vec{y}$
}

\Return{$\vec{y}$}
\end{algorithm}

\begin{algorithm}
\caption{$\cA_\had$, analyzer for histograms}
\label{alg:a-had}

\KwIn{$\vec{y} \in ([2d]^\tau)^{n(k+1)}$}
\KwOut{$\vec{z}\in\R^d$}

\For{$j\in[d]$}{
    $c_j\gets 0$
    
    \For{$(a_1,\dots,a_\tau)\in\vec{y}$}{
        \If{every $a_1,\dots,a_\tau \in \{ \hat{j} ~|~ h_{j,\hat{j}} = 1 \}$}{
            $c_j \gets c_j+1$
        }
    }
    
    $z_j\gets \tfrac{1}{n}\cdot \tfrac{1}{1-2^{-\tau}} \cdot (c_j - n(k+1)\cdot 2^{-\tau})$
}

\Return{$\vec{z}$}
\end{algorithm}

\begin{clm*}[Restatement of \ref{clm:had-attack}]
Choose $k,\tau$ as in Theorem \ref{thm:had}. If there is a coalition of $m<n$ corrupt users $M\subset [n]$, then for any target value $j\in[d]$ there is an input $\vec{x}$ such that $\cP_\flip$ produces an estimate of $\hist_j(\vec{x})$ with bias $\tfrac{m}{n} \cdot (k+1) = \Omega( \tfrac{m}{n}\cdot \tfrac{1}{\eps^2} \log \tfrac{1}{\eps\delta})$.
\end{clm*}
\begin{proof}
The attack is simple: given the target $j$, each corrupt user samples $k+1$ values i.i.d. from $\{\hat{j} ~|~ h_{j,\hat{j}} =1 \}$ in lieu of running $\cR_\had$. Now consider an input $\vec{x}$ such that $\hist_j(\vec{x})=0$. For $z^{\mathrm{cor}}_j$ computed by $\cP_\had$ under attack by $m$ corrupt users, we will argue that $\ex{}{z^{\mathrm{cor}}_j} = \tfrac{m}{n} \cdot (k+1)$.

We require some notation. Let $b_{i,r}$ be the bit that indicates if the $r$-th message produced by user $i$ will have the property that each element belongs to $\{ \hat{j} ~|~ h_{j,\hat{j}} = 1 \}$. Note that $c_j = \sum_{i\in[n]} \sum_{r\in[k+1]} b_{i,r}$. We will use the superscripts ``${\mathrm{hon}}$'' and ``${\mathrm{cor}}$'' to denote random variables from honest and corrupted executions, respectively.

\begin{align*}
\ex{}{z^{\mathrm{cor}}_j} &= \frac{1}{n}\cdot \frac{1}{1-2^{-\tau}} \cdot (\ex{}{c^{\mathrm{cor}}_j} - n(k+1)\cdot 2^{-\tau}) \\
    &= \frac{1}{n}\cdot \frac{1}{1-2^{-\tau}} \cdot \left(\sum_{i\in[n]} \sum_{r\in[k+1]} \ex{}{b^{\mathrm{cor}}_{i,r}} - n(k+1)\cdot 2^{-\tau} \right) \\
    &= \frac{1}{n}\cdot \frac{1}{1-2^{-\tau}} \cdot \left(\sum_{i\notin M} \sum_{r\in[k+1]} \ex{}{b^{\mathrm{cor}}_{i,r}}  + \sum_{i\in M} \sum_{r\in[k+1]} \ex{}{b^{\mathrm{cor}}_{i,r}} - n(k+1)\cdot 2^{-\tau} \right) \\
    &= \frac{1}{n}\cdot \frac{1}{1-2^{-\tau}} \cdot \left(\sum_{i\notin M} \sum_{r\in[k+1]} \ex{}{b^{\mathrm{hon}}_{i,r}}  + \sum_{i\in M} \sum_{r\in[k+1]} \ex{}{b^{\mathrm{cor}}_{i,r}} - n(k+1)\cdot 2^{-\tau} \right) \\
    &= \frac{1}{n}\cdot \frac{1}{1-2^{-\tau}} \cdot \left((n-m)(k+1)\cdot 2^{-\tau} + \sum_{i\in M} \sum_{r\in[k+1]} \ex{}{b^{\mathrm{cor}}_{i,r}} - n(k+1)\cdot 2^{-\tau} \right) \stepcounter{equation} \tag{\theequation} \label{eq:prior-work}\\
    &= \frac{1}{n}\cdot \frac{1}{1-2^{-\tau}} \cdot \left((n-m)(k+1)\cdot 2^{-\tau} + m(k+1) - n(k+1)\cdot 2^{-\tau} \right) \stepcounter{equation} \tag{\theequation} \label{eq:attack-def}\\
    &=\frac{m}{n}\cdot (k+1)
\end{align*}
\eqref{eq:prior-work} comes from analysis done in \cite{GGKPV19}. \eqref{eq:attack-def} is immediate from the definition of the attack.
\end{proof}

%% file: appendix/amplification.tex
\section{Histogram Protocol via Privacy Amplification}
\label{apdx:amplification}
In this appendix, we will consider the variant of $\cR_\flip$ where there are no messages from fabricated users. The privacy analysis is performed using the amplification-by-shuffling result by Feldman et al. \cite{FMT20}.

\begin{thm}
\label{thm:amplification-variant}
Fix any $\eps \leq 4$, $\delta < 1$, and $k=0$. For any $n > \max \left( \tfrac{1024}{\eps^2}\ln \frac{4}{\delta}, 6\ln20d \right)$, there is a choice of parameter $q < 1/3$ such that the protocol $\cP_\flip=(\cR_\flip,\cA_\flip)$ has the following properties
\begin{enumerate}[a.]
    \item $\cP_\flip$ is $(\eps,\delta)$-shuffle private
    \item For any $\vec{x}\in \cX^n_d$, $\cP_\flip(\vec{x})$ reports a vector $\vec{z}$ such that the maximum error with respect to $\hist(\vec{x})$ is $$\norm{\vec{z} - \hist(\vec{x})}_\infty < \max\left( \frac{24}{n^{3/4}\sqrt{\eps}} \left(\ln\frac{4}{\delta}\right)^{1/4} \sqrt{ \ln 20d} ,~ \frac{6}{n} \ln 20d \right)$$ with 90\% probability.
\end{enumerate}
\end{thm}

We first restate the amplification lemma from \cite{FMT20} using our notation and variant of the model.
\begin{lem}
\label{lem:amplification}
Fix any $\delta \in (0,1)$, $n\in \N$, and $\eps_L \leq \ln ( n / 16 \ln (2/\delta))$. If $\cR:\cX\to\cY$ is $\eps_L$-differentially private then $(\cS\circ\cR^n)$ is $(\eps_S,\delta)$-differentially private, where $$\eps_S = 8 \cdot \frac{e^{\eps_L} -1}{e^{\eps_L} +1} \cdot \left( \sqrt{\frac{e^{\eps_L} \ln (4/\delta) }{n}} + \frac{e^{\eps_L}}{n} \right).$$
\end{lem}

A corollary of this lemma is that when the target privacy parameter $\eps_S$ is sufficiently small, there is always some choice of privacy parameter $\eps_L$ for $\cR$ and some threshold for $n$ above which the shuffle protocol is $(\eps_S,\delta)$-shuffle private. More precisely,
\begin{coro}
\label{coro:amplification}
Fix any $\eps_S \leq 4$ and $\delta \in (0,1)$. If $n > \tfrac{256}{\eps^2_S} \ln(4/\delta)$ and $\cR:\cX\to\cY$ is $\eps_L$-differentially private for $\eps_L \leq \ln(\eps^2_S n / 256\ln(4/\delta))$, then $(\cS\circ\cR^n)$ is $(\eps_S,\delta)$-differentially private.
\end{coro}
\begin{proof}
Because $\eps_S$ is sufficiently small, $\eps_L$ satisfies the condition under which Lemma \ref{lem:amplification} holds: $(\cS\circ\cR^n)$ is $(\eps,\delta)$-differentially private, where
\begin{align*}
\eps &= 8 \cdot \frac{e^{\eps_L} -1}{e^{\eps_L} +1} \cdot \left( \sqrt{\frac{e^{\eps_L} \ln (4/\delta) }{n}} + \frac{e^{\eps_L}}{n} \right) \\
    &\leq 8 \cdot \left( \sqrt{\frac{e^{\eps_L} \ln (4/\delta) }{n}} + \frac{e^{\eps_L}}{n} \right) \\
    &\leq 8 \cdot \left( \frac{\eps_S}{16} + \frac{\eps^2_S}{256\ln(4/\delta)} \right) \tag{Bound on $\eps_L$} \\
    &= \frac{\eps_S}{2} + \frac{\eps_S}{2} \cdot \frac{\eps_S}{16\ln(4/\delta)} \\
    &\leq \eps_S
\end{align*}
The final inequality follows from our bound on $\eps_S$.
\end{proof}

Now we find a value of $q$ to ensure $\cP_\flip$ satisfies $\eps_L$-local privacy:
\begin{clm}
\label{clm:local-privacy}
For any $\eps_L > 0$, if $k\gets 0$ and $q\gets 1/(e^{\eps_L/2}+1)$ then the randomizer $\cR_\flip$ is $\eps_L$-differentially private.
\end{clm}

Theorem \ref{thm:amplification-variant} follows from Corollary \ref{coro:amplification}, Corollary \ref{coro:acc-max}, and Claim \ref{clm:local-privacy}.
\begin{proof}[Proof of Theorem \ref{thm:amplification-variant}]
We choose $\eps_L \gets \ln(\eps^2 n / 256\ln(4/\delta))$ and $q\gets \max( 1/(e^{\eps_L/2}+1), \tfrac{1}{n} \ln 20d)$. By substitution, we have that the following holds with probability $9/10$:
\begin{align*}
\norm{\cP_\flip(\vec{x}) - \hist(\vec{x})}_\infty &< 2 \sqrt{\frac{1}{n} \cdot q \ln 20d } \cdot \left(\frac{1}{1-2q}\right) \tag{$k=0,q>0$}\\
    &= 2 \sqrt{\frac{1}{n}\cdot \max\left( \frac{1}{e^{\eps_L/2}+1} ,~ \frac{1}{n} \ln 20d \right) \ln 20d } \cdot \left(\frac{1}{1-2q}\right) \\
    &< 2 \sqrt{\frac{1}{n}\cdot \max\left( \frac{16\sqrt{\ln(4/\delta)}}{\eps \sqrt{n}} ,~ \frac{1}{n} \ln 20d \right) \ln 20d } \cdot \left(\frac{1}{1-2q}\right) \\
    &= \max\left( \frac{8}{n^{3/4}\sqrt{\eps}} \left(\ln\frac{4}{\delta}\right)^{1/4} \sqrt{ \ln 20d} ,~ \frac{2}{n} \ln 20d \right)\cdot \left(\frac{1}{1-2q}\right)
\end{align*}
Given that $n$ is sufficiently large, we conclude $q < 1/3$. The Theorem follows by substitution.
\end{proof}